%% file: full-version.tex
\documentclass[runningheads]{llncs}
\usepackage{graphicx}

\usepackage[utf8]{inputenc}
\usepackage{ltl}
\usepackage{tikz}
\usetikzlibrary{arrows,arrows.meta,automata,shapes,shapes.multipart,positioning,calc}
\usepackage{mathtools}
\usepackage{subcaption}
\usepackage{colortbl}
\usepackage{xspace}

\usepackage[stable]{footmisc}

\usepackage{hyperref}

\usepackage[capitalise,nameinlink]{cleveref}

\include{macros}

\begin{document}
\title{Compositional Synthesis of Modular Systems (Full Version)\thanks{This work was partially supported by the German Research Foundation~(DFG) as part of the Collaborative Research Center ``Foundations of Perspicuous Software Systems'' (TRR 248, 389792660), and by the European Research Council (ERC) Grant OSARES (No. 683300).}\footnote[3]{This is an extended version of~\cite{FinalVersion}.}
}
\titlerunning{Compositional Synthesis of Modular Systems}
%
\author{Bernd Finkbeiner \and Noemi Passing}
\authorrunning{B.\ Finkbeiner and N.\ Passing}
%
\institute{CISPA Helmholtz Center for Information Security, Saarbrücken, Germany\\
\email{\{finkbeiner,noemi.passing\}@cispa.de}}
\maketitle              
\begin{abstract}
  Given the advances in reactive synthesis, it is a natural next step to consider more complex multi-process systems. Distributed synthesis, however, is not yet scalable. Compositional approaches can be a game changer. Here, the challenge is to decompose a given specification of the global system behavior into requirements on the individual processes. In this paper, we introduce a compositional synthesis algorithm that, for each process, constructs, in addition to the implementation, a certificate that captures the necessary interface between the processes. The certificates then allow for constructing separate requirements for the individual processes. By bounding the size of the certificates, we can bias the synthesis procedure towards solutions that are desirable in the sense that the assumptions between the processes are small. Our experimental results show that our approach is much faster than standard methods for distributed synthesis as long as reasonably small certificates exist.

\end{abstract}

\input{introduction}


\input{motivating_example}


\input{preliminaries}


\input{compositional_certification}


\input{computing_guarantees}


\input{relevant_processes}


\input{synthesizing_certificates}


\input{experiments}


\input{conclusions}

%
%
%
\bibliographystyle{splncs04}
\bibliography{bib}


\appendix
\input{appendix}

\end{document}

%% file: macros.tex
\newcommand{\Next}{\LTLnext}
\newcommand{\Globally}{\LTLsquare}
\newcommand{\Eventually}{\LTLdiamond}
\newcommand{\Until}{\LTLuntil}

\newcommand{\true}{\mathit{true}}

\newcommand{\ie}{i.e.\@\xspace}
\newcommand{\eg}{e.g.\@\xspace}

\DeclareMathOperator{\pc}{||}

\DeclarePairedDelimiter\myVec{\langle}{\rangle}

\newcommand{\constraintSystem}[3]{\mathcal{C}_{#1,#2,#3}}

\newcommand{\gb}[1]{\psi_{#1}}
\newcommand{\GB}[1]{\Psi_{#1}}
\newcommand{\relGB}[1]{\Psi^\mathcal{R}_{#1}}
\newcommand{\relGBPrime}[1]{\Psi'_{#1,\mathcal{R}}}

\newcommand{\guarTrans}[1]{\mathcal{T}^G_{#1}}

\newcommand{\propositions}[1]{\operatorname{prop}(#1)}

\newcommand{\extend}[2]{\operatorname{extend}(#1,#2)}
\newcommand{\restrict}[2]{\operatorname{restrict}(#1,#2)}

\newcommand{\pref}[2]{#1_{.. #2}}

\newcommand{\simRelStG}[1]{R^{s\rightarrow g}_{#1}}
\newcommand{\simRelGtS}[2]{R^{g\rightarrow s}_{#1,#2}}

\newcommand{\comp}{\mathit{comp}}
\newcommand{\validHistory}[2]{\mathcal{H}^{#1}_{#2}}

\newcommand{\inputs}[1]{{I_{#1}}}
\newcommand{\outputs}[1]{{O_{#1}}}
\newcommand{\variables}[1]{{V_{#1}}}
\newcommand{\inpFunc}{{I}}
\newcommand{\outFunc}{{O}}
\newcommand{\guarOutputs}[1]{{O^G_{#1}}}
\newcommand{\guarVariables}[1]{{V^G_{#1}}}

\newcommand{\allInputs}{\mathit{inp}}
\newcommand{\allOutputs}{\mathit{out}}
\newcommand{\associatedOutputs}[1]{\outputs{\relevantProcesses{#1}}}
\newcommand{\envOutputs}{{O_\mathit{env}}}

\newcommand{\env}{\mathit{env}}
\newcommand{\sysProc}{P^-\!}

\newcommand{\inp}[1]{\boldsymbol{#1}}
\newcommand{\out}[1]{\boldsymbol{#1}}

\newcommand{\relevantProcesses}[1]{\mathcal{R}_{#1}}
\newcommand{\relevantProcessesFunc}{\mathcal{R}}

\newcommand{\satresp}[1]{\models_{#1}}

\newcommand{\simresp}[1]{\preceq}

\newcommand{\statesStrat}[1]{T_{#1}}
\newcommand{\statesSpec}[1]{Q_{#1}}
\newcommand{\statesGuar}[1]{G_{#1}}

\newcommand{\validInput}[3]{\operatorname{valid}_{#1}(#2,#3)}

\newcommand{\transStrat}[4]{\tau^{#1}_{#2,#3,#4}}
\newcommand{\transSpec}[4]{\delta^{#1}_{#2,#3,#4}}
\newcommand{\transSpecSugar}[5]{\delta^{#1}_{#2,#3,#4,#5}}
\newcommand{\transGuar}[4]{{\tau}^{G,#1}_{#2,#3,#4}}

\newcommand{\outputStrat}[4]{o^{#1}_{#2,#4}}

\newcommand{\outputGuar}[4]{{o}^{G,#1}_{#2,#4}}

\newcommand{\reach}[3]{\lambda^{#1,\mathbb{B}}_{#2,#3}}
\newcommand{\bound}[3]{\lambda^{#1,\#}_{#2,#3}}
\newcommand{\greaterBound}[1]{\triangleright_{#1}}

\newcommand{\simStratToGuar}[3]{{\preceq}^{S \rightarrow G,#1}_{#2,#3}}
\newcommand{\simGuarToStrat}[4]{{\preceq}^{G \rightarrow S,#1,#2}_{#3,#4}}

%% file: introduction.tex
\section{Introduction}

In the last decade, there have been breakthroughs in terms of realistic applications and practical tools for reactive synthesis, demonstrating that concentrating on \emph{what} a system should do instead of \emph{how} it should be done is feasible.
A natural next step is to consider complex multi-process systems. For distributed systems, though, there are no scalable tools that are capable of automatically synthesizing strategies from formal specifications for arbitrary system architectures.

For the scalability of verification algorithms, compositionality, \ie, breaking down the verification of a complex system into several smaller tasks over individual components, has proven to be a key technique~\cite{Compos97}. For synthesis, however, developing compositional approaches is much more challenging: In practice, an individual process can rarely guarantee the satisfaction of the specification alone. Typically, there exist input sequences that prevent a process from satisfying the specification. The other processes in the system then ensure that these sequences are not produced. Thus, a process needs information about the strategies of the other processes to be able to satisfy the specification. Hence, distributed synthesis cannot easily be broken down into tasks over the individual processes.

In this paper, we introduce a compositional synthesis algorithm addressing this problem by synthesizing additional \emph{guarantees on the behavior} of every process. These guarantees, the so-called \emph{certificates}, then provide essential information for the individual synthesis tasks: A strategy is only required to satisfy the specification if the other processes do not deviate from their guaranteed~behavior. This allows for considering a process independent of the other processes' strategies.
Our algorithm is an extension of bounded synthesis~\cite{FinkbeinerS13} that incorporates the search for certificates into the synthesis task for the strategies.

The benefits of synthesizing additional certificates are threefold. 
First, it \emph{guides the synthesis procedure}: Bounded synthesis searches for strategies up to a given size. Beyond that, our algorithm introduces a bound on the size of the certificates. Hence, it bounds the size of the interface between the processes and thus the size of the assumptions made by them. By starting with small bounds and by only increasing them if the specification is unrealizable for the~given bounds, the algorithm restricts synthesis to search for solutions with small interfaces.

Second, the certificates \emph{increase the understandability} of the synthesized solution: It is challenging to recognize the interconnections in a distributed system. The certificates capture which information a process needs about the behavior of the other processes to be able to satisfy the specification, immediately encapsulating the system's interconnections.
Furthermore, the certificates abstract from behavior that is irrelevant for the satisfaction of the specification. This allows for analyzing the strategies locally without considering the whole system's behavior.

Third, synthesizing certificates \emph{enables modularity} of the system: The strategies only depend on the certificates of the other processes, not on their particular strategies. As long as the processes do not deviate from their certificates, the parallel composition of the strategies satisfies the specification. Hence, the certificates form a contract between the processes. After defining the contract, the strategies can be exchanged safely with other ones that respect the contract. Thus, strategies can be adapted flexibly without synthesizing a solution for the whole system again if requirements that do not affect the contract change.

We introduce two representations of certificates, as LTL formulas and as labeled transition systems. We show soundness and completeness of our certifying synthesis algorithm for both of them. 
Furthermore, we present a technique for determining \emph{relevant processes} for each process. This allows us to reduce the number of certificates that a process has to consider to satisfy the specification while maintaining soundness and completeness.
Focusing on the representation of certificates as transition systems, we present an algorithm for synthesizing certificates that is based on a reduction to a SAT constraint system.

We implemented the algorithm and compared it to an extension~\cite{Baumeister17} of the~syn\-the\-sis tool BoSy~\cite{FaymonvilleFT17} to distributed systems and to a compositional synthesis algorithm based on dominant strategies~\cite{DammF14}. 
The results clearly demonstrate the advantage of synthesizing certificates: If solutions with a small interface between the processes exist, our algorithm outperforms the other synthesis tools significantly. Otherwise, the overhead of synthesizing additional guarantees is small.

\input{related_work}

%% file: related_work.tex
\noindent
\textbf{Related Work:}
There are several approaches to compositional synthesis for mo\-no\-li\-thic systems~\cite{KupfermanPV06,FiliotJR10,KuglerS09,FinkbeinerP20,FinkbeinerGP21}.
As we are considering distributed systems, we focus on distributed synthesis algorithms. Assume-guarantee synthesis~\cite{ChatterjeeH07} is closest to our approach.
There, each process provides a guarantee on its own behavior and makes an assumption on the behavior of the other processes.
If there is a strategy for each process that satisfies the specification under the hypothesis that the other processes respect the assumption, and if its guarantee implies the assumptions of the other processes, a solution for the whole system is found.
In contrast to our approach, most assume-guarantee synthesis algorithms~\cite{ChatterjeeH07,BrenguierRS17,BloemCJK15,AlurMT15} either rely on the user to provide the assumptions or require that a strategy profile on which the strategies can synchronize is constructed prior to synthesis.

A recent extension of assume-guarantee synthesis~\cite{MajumdarMSZ20} algorithmically synthesizes assume-guarantee contracts for each process.
In contrast to our approach, the guarantees do not necessarily imply the assumptions of the other processes. Thus, the algorithm needs to iteratively refine assumptions and guarantees until a valid contract is found.
This iteration is circumvented in our algorithm since only assumptions that are guaranteed by the other processes are used.

Using a weaker winning condition for synthesis, remorse-free dominance~\cite{DammF11}, avoids the explicit construction of assumptions and guarantees~\cite{DammF14}. The assumptions are implicit, but they do not always suffice. Thus, although a dependency analysis of the specification allows for solutions for further, more interconnected systems and specifications~\cite{FinkbeinerP20}, compositional solutions do not always exist.

%% file: motivating_example.tex
\section{Running Example}\label{sec:motivating_example}

In many modern factories, autonomous robots are a crucial component in the production line. 
The correctness of their implementation is essential and therefore they are a natural target for synthesis.
Consider a factory with two robots that carry production parts from one machine to another.
In the factory, there is a crossing that is used by both robots.
The robots are required to prevent a crash: $\varphi_\mathit{safe} := \Globally \neg ( (\mathit{at\_crossing}_1 \land \Next \mathit{go}_1) \land (\mathit{at\_crossing}_2 \land \Next \mathit{go}_2))$, where $\mathit{at\_crossing}_i$ is an input variable denoting that robot $r_i$ arrived at the crossing, and $\mathit{go}_i$ is an output variable of robot $r_i$ denoting that $r_i$ moves ahead.
Moreover, both robots need to cross the intersection at some point in time after arriving there: $\varphi_{\mathit{cross}_i} := \Globally (\mathit{at\_crossing}_i \rightarrow \Next \Eventually \mathit{go}_i)$.
In addition to these requirements, both robots have further objectives $\varphi_{\mathit{add}_i}$ that are specific to their area of application. For instance, they may capture which machines have to be approached.

None of the robots can satisfy $\varphi_\mathit{safe} \land \varphi_{\mathit{cross}_i}$ alone: The crossing needs to be entered eventually by $r_i$ but no matter when it is entered, $r_j$ might enter it at the same time. Thus, strategies cannot be synthesized individually without information on~the other robot's behavior.
Due to~$\varphi_{\mathit{add}_i}$, the parallel composition of the strategies can be large and complex. Hence, understanding why the overall specification is met and recognizing the individual strategies is challenging.

If both robots commit to their behavior at crossings, a robot $r_i$ can satisfy $\varphi_\mathit{safe}\land\varphi_{\mathit{cross}_i}$ individually since it is allowed to assume that the other robot does not deviate from its guaranteed behavior, the so-called certificate.
For instance, if $r_2$ commits to always giving priority to $r_1$, entering the crossing regardless of~$r_2$ satisfies $\varphi_\mathit{safe}\land\varphi_{\mathit{cross}_1}$ for $r_1$.
If $r_1$ guarantees to not block crossings, $r_2$ can satisfy $\varphi_\mathit{safe}\land\varphi_{\mathit{cross}_2}$ as well.
Hence, if both robots can satisfy the whole part of the specification that affects them, \ie, $\varphi_i = \varphi_\mathit{safe} \land \varphi_{\mathit{cross}_i} \land \varphi_{\mathit{add}_i}$, under the assumption that the other robot sticks to its certificate, then the parallel composition of their strategies satisfies the whole specification. Furthermore, we then know that the robots do not interfere in any other situation.
Thus, the certificates provide insight in the required communication of the robots.

Moreover, when analyzing the strategy $s_i$ of $r_i$, only taking $r_j$'s certificate into account abstracts away $r_j$'s behavior aside from crossings. This allows us to focus on the relevant aspects of $r_j$'s behavior for $r_i$, making it significantly easier to understand why $r_i$'s strategy satisfies~$\varphi_i$. 
Lastly, the certificates form a contract of safe behavior at crossings: If $r_i$'s additional objectives change, it suffices to synthesize a new strategy for~$r_i$. Provided $r_i$ does not change its behavior at crossings, $r_j$'s strategy can be left unchanged.

%% file: preliminaries.tex
\section{Preliminaries}

\paragraph{Notation.} In the following, we denote the prefix of length $t$ of an infinite word $\sigma = \sigma_1 \sigma_2 \dots \in (2^V)^\omega$ by $\pref{\sigma}{t} := \sigma_{1} \dots \sigma_{t}$. Moreover, for a set $X$ and an infinite word $\sigma = \sigma_1 \sigma_2 \dots \in (2^V)^\omega$, we define $\sigma \cap X = (\sigma_1 \cap X)(\sigma_2 \cap X)\dots\in (2^X)^\omega$.

\paragraph{LTL.} 
Linear-time temporal logic~(LTL)~\cite{Pnueli77} is a specification language for linear-time properties. 
Let $\Sigma$ be a finite set of atomic propositions and let $a \in \Sigma$. The syntax of LTL is given by
$ \varphi, \psi ::= a ~ | ~ \neg \varphi ~ | ~ \varphi \lor \psi ~ | ~ \varphi \land \psi ~ | ~ \Next \varphi ~ | ~ \varphi \Until \psi$.
We define $\Eventually \varphi = \true \Until \varphi$, and $\Globally \varphi = \neg \Eventually \neg \varphi$ and use the standard semantics.
The language $\mathcal{L}(\varphi)$ of a formula $\varphi$ is the set of infinite words that satisfy $\varphi$.
The atomic propositions in $\varphi$ are denoted by $\propositions{\varphi}$.
We represent a formula $\varphi = \xi_1 \land \dots \land \xi_k$ also by the set of its conjuncts, \ie, $\varphi = \{\xi_1, \dots, \xi_k\}$.

\paragraph{Automata.}
A universal co-Büchi automaton $\mathcal{A} = (Q,q_0,\delta,F)$ over a finite alphabet $\Sigma$ consists of a finite set of states $Q$, an initial state $q_0 \in Q$, a transition relation $\delta: Q \times 2^\Sigma \times Q$, and a set $F \subseteq Q$ of rejecting states.
For an infinite word $\sigma = \sigma_0\sigma_1 \dots \in (2^\Sigma)^\omega$, a run of $\sigma$ on $\mathcal{A}$ is an infinite sequence $q_0 q_1 \dots \in Q^\omega$ of states with $(q_i,\sigma_i,q_{i+1}) \in \delta$ for all $i \geq 0$. 
A run is accepting if it contains only finitely many visits to rejecting states. $\mathcal{A}$ accepts a word $\sigma$ if all runs of $\sigma$ on $\mathcal{A}$ are accepting.
The language $\mathcal{L}(\mathcal{A})$ of $\mathcal{A}$ is the set of all accepted words.
An LTL specification $\varphi$ can be translated into an equivalent universal co-Büchi automaton $\mathcal{A}_\varphi$, \ie, with $\mathcal{L}(\varphi) = \mathcal{L}(\mathcal{A}_\varphi)$, with a single exponential blow up~\cite{KupfermanV05}.

\paragraph{Architectures.}
An architecture is a tuple $A=(P, V, \inpFunc, \outFunc)$, where $P$ is a set of processes consisting of the environment process $\env$ and a set of $n$ system processes $\sysProc = P \setminus\{env\}$, $V$ is a set of variables, $\inpFunc = \myVec{I_1, \dots, I_n}$ assigns a set $\inputs{j} \subseteq V$ of input variables to each system process $p_j$, and $\outFunc = \myVec{O_\env, O_1, \dots O_n}$ assigns a set $\outputs{j} \subseteq V$ of output variables to each process $p_j$.
For all $p_j, p_k \in \sysProc$ with $j \neq k$, we have $\inputs{j} \cap \outputs{j} = \emptyset$ and $\outputs{j} \cap \outputs{k} = \emptyset$.
The variables $\variables{j}$ of $p_j \in \sysProc$ are its inputs and outputs, \ie, $\variables{j} = \inputs{j} \cup \outputs{j}$.
The variables $V$ of the whole system are defined by $V = \bigcup_{p_j \in \sysProc} \variables{j}$. 
We define $\allInputs = \bigcup_{p_j \in \sysProc} \inputs{j}$ and $\allOutputs = \bigcup_{p_j \in \sysProc} \outputs{j}$.
An architecture is called distributed if $|\sysProc| \geq 2$ and monolithic otherwise. 
In the remainder of this paper, we assume that a distributed architecture is given.

\paragraph{Transition Systems.}
Given sets $I$ and $O$ of input and output variables, a Moore transition system (TS) $\mathcal{T} = (T,t_0,\tau,o)$ consists of a finite set of states~$T$, an initial state $t_0$, a transition function $\tau: T \times 2^I \rightarrow T$, and a labeling function $o: T \rightarrow 2^O$.
For an input sequence $\gamma = \gamma_0 \gamma_1 \dotsc \in (2^{I})^\omega$, $\mathcal{T}$ produces a path $\pi = (t_0 , \gamma_0 \cup o(t_0)) (t_1 , \gamma_1 \cup o(t_1)) \dotsc \in (T \times 2^{I \cup O})^\omega$, where $(t_j, \gamma_j,t_{j+1}) \in \tau$.
The projection of a path to the variables is called trace.
The parallel composition of two TS $\mathcal{T}_1 = (T_1,t^1_0,\tau_1,o_1)$, $\mathcal{T}_2 = (T_2,t^2_0,\tau_2,o_2)$, is a TS $\mathcal{T}_1 \pc \mathcal{T}_2 = (T,t_0,\tau,o)$~with $T = T_1 \times T_2$, $t_0=(t^1_0,t^2_0)$, 
$\tau(\!(t,t'),\inp{i}) \!=\! (\tau_1(t,(\inp{i}_1\cup o_2(t')\!) \cap \inputs{1}),\tau_2(t',(\inp{i}_2\cup o_1(t)\!)\cap\inputs{2})\!)$, and 
$o((t,t')) = o_1(t) \cup o_2(t')$.
A TS $\mathcal{T}_1 = (T_1,t^1_0,\tau_1,o_1)$ over $I$ and~$O_1$ \emph{simulates} $\mathcal{T}_2 = (T_2,t^2_0,\tau_2,o_2)$ over $I$ and $O_2$ with $O_1 \subseteq O_2$, denoted $\mathcal{T}_2 \preceq \mathcal{T}_1$, if there is a simulation relation $R: T_2 \times T_1$ with $(t^2_0,t^1_0)\in R$, $\forall (t_2,t_1) \in R.~ o(t_2) \cap O_1 = o(t_1)$, and $\forall t'_2 \in T_2. \forall \inp{i} \in 2^I.~ (\tau_2(t_2,\inp{i}) = t'_2) \rightarrow (\exists t'_1 \in T_1.~\tau_1(t_1,\inp{i})=t'_1 \land (t'_2,t'_1) \in R)$.

\paragraph{Strategies.}
We model a strategy $s_i$ of $p_i\in\sysProc$ as a Moore transition system~$\mathcal{T}_i$ over $\inputs{i}$ and $\outputs{i}$. The trace produced by $\mathcal{T}_i$ on $\gamma \in (2^{\inputs{i}})^\omega$ is called the \emph{computation} of $s_i$ on $\gamma$, denoted $\comp(s_i,\gamma)$. For an LTL formula $\varphi$ over $V$, $s_i$ satisfies $\varphi$, denoted $s_i \models \varphi$, if $\comp(s,\gamma) \cup \gamma' \models \varphi$ holds for all $\gamma \in (2^{\inputs{i}})^\omega$, $\gamma' \in (2^{V\setminus\variables{i}})^\omega$.

\paragraph{Synthesis.}
For a specification $\varphi$, synthesis derives strategies $s_1, \dots, s_n$ for the system processes such that $s_1 \pc \dots \pc s_n \models \varphi$ holds. 
If such strategies exist,~$\varphi$ is realizable in the architecture. Bounded synthesis~\cite{FinkbeinerS13} additionally bounds the size of the strategies.
The search for strategies is encoded into a constraint system that is satisfiable if, and only if, $\varphi$ is realizable for the bound.
There are~SMT,~SAT, QBF, and DQBF encodings for monolithic~\cite{FaymonvilleFRT17} and distributed~\cite{Baumeister17} architectures.

%% file: compositional_certification.tex
\section{Compositional Synthesis with Certificates}\label{sec:certifying_synthesis}

In this section, we describe a sound and complete compositional synthesis algorithm for distributed systems. The main idea is to synthesize strategies for the system processes individually. Hence, in contrast to classical distributed synthesis, where strategies $s_1, \dots, s_n$ are synthesized such that $s_1 \pc \dots \pc s_n \models \varphi$ holds, we require that $s_i \models \varphi_i$ holds for all system processes $p_i \in \sysProc$. Here, $\varphi_i$ is a subformula of $\varphi$ that, intuitively, captures the part of $\varphi$ that affects $p_i$.
As long as $\varphi_i$ contains all parts of~$\varphi$ that restrict the behavior of $s_i$, the satisfaction of~$\varphi$ by the parallel composition of all strategies is guaranteed.
Computing specification decompositions is not the main focus of this paper; in fact, our algorithm can be used with any decomposition that fulfills the above requirement. There is work on obtaining small subspecifications, \eg, \cite{FinkbeinerGP21}, we, however, use an easy decomposition algorithm in the remainder of this paper for simplicity:

\begin{definition}[Specification Decomposition]\label{def:spec_decomposition}
	Let $\varphi = \xi_1 \land \dots \land \xi_k$ be an LTL formula. The \emph{decomposition of~$\varphi$} is a vector $\myVec{ \varphi_1, \dots, \varphi_n }$ of LTL formulas with $\varphi_i = \{ \xi_j \mid \xi_j \in \varphi \,\land\, (\propositions{\xi_j} \cap \outputs{i} \neq \emptyset \,\lor\, \propositions{\xi_j} \cap \allOutputs = \emptyset) \}$.
\end{definition}

Intuitively, the subspecification~$\varphi_i$ contains all conjuncts of $\varphi$ that contain outputs of~$p_i$ as well as all input-only conjuncts.
In the remainder of this paper, we assume that both $\propositions{\varphi} \subseteq V$ and $\mathcal{L}(\varphi) \in (2^V)^\omega$ hold for all specifications~$\varphi$.
Then, every atomic proposition occurring in a for\-mu\-la $\varphi$ is an input or output of at least one system process and thus $\bigwedge_{p_i\in\sysProc} \varphi_i = \varphi$ holds.

Although we decompose the specification, a process $p_i$ usually cannot guarantee the satisfaction of $\varphi_i$ alone; rather, it depends on the cooperation of the other processes. For instance, robot $r_1$ from \Cref{sec:motivating_example} cannot guarantee that no crash will occur when entering the crossing since~$r_2$ can enter it at the same point in time. Thus, we additionally synthesize a \emph{guarantee on the behavior} of each process, the so-called \emph{certificate}. The certificates then provide essential information to the processes:~If~$p_i$~commits to a certificate, the other processes can rely on~$p_i$'s strategy to not deviate from this behavior. In particular, the strategies only need to satisfy the specification as long as the other processes stick to their certificates.
Thus, a process is not required to react to \emph{all} behaviors of the other processes but only to those that truly occur when the processes interact.

In this section, we represent the certificate of $p_i \in \sysProc$ by an LTL formula $\gb{i}$. For instance, robot $r_2$ may guarantee to always give priority to $r_1$ at crossings, yielding the certificate $\gb{2} = \Globally ((\mathit{at\_crossing}_1 \land \mathit{at\_crossing}_2) \rightarrow \Next \neg \mathit{go}_2)$.
Since~$r_1$ can assume that~$r_2$ does not deviate from its certificate $\psi_2$, a strategy for~$r_1$ that enters crossings regardless of~$r_2$ satisfies $\varphi_\mathit{safe} \land \varphi_{\mathit{cross}_1}$.

To ensure that $p_i$ does not deviate from its own certificate, we require its strategy $s_i$ to satisfy the LTL formula $\psi_i$ describing it. To model that $s_i$ only has to satisfy its specification if the other processes stick to their certificates, it has to satisfy $\GB{i} \rightarrow \varphi_i$, where $\GB{i} = \{ \gb{j} \mid p_j \in \sysProc\setminus\{p_i\}\}$, \ie, $\GB{i}$ is the conjunction of the certificates of the other processes.
Using this, we define certifying synthesis:

\begin{definition}[Certifying Synthesis]
	Let $\varphi$ be an LTL formula with decomposition $\myVec{ \varphi_1, \dots, \varphi_n }$. 
	\emph{Certifying synthesis} derives strategies $s_1,\dots,s_n$ and LTL certificates $\gb{1},\dots,\gb{n}$ for the system processes such that $s_i \models \gb{i} \land (\GB{i} \rightarrow \varphi_i)$ holds for all $p_i \in \sysProc$, where $\GB{i} = \{ \gb{j} \mid p_j \in \sysProc\setminus\{p_i\}\}$.
\end{definition}

Classical distributed synthesis algorithms reason \emph{globally} about the satisfaction of the full specification by the parallel composition of the synthesized strategies.
Certifying synthesis, in contrast, reasons \emph{locally} about the satisfaction of the subspecifications for the individual processes, \ie, without considering the parallel composition of the strategies. This greatly improves the understandability of the correctness of synthesized solutions since we are able to consider the strategies separately.
Furthermore, local reasoning is still sound and complete:

\begin{theorem}[Soundness and Completeness]\label{thm:soundness_completeness_certifying_synthesis}
	Let $\varphi$ be an LTL formula and let $\mathcal{S} = \myVec{s_1, \dots, s_n}$ be a vector of strategies for the system processes. There exists a vector $\Psi = \myVec{\gb{1}, \dots, \gb{n}}$ of LTL certificates such that $(\mathcal{S},\Psi)$ is a solution of certifying synthesis for $\varphi$ if, and only if $s_1 \pc \dots \pc s_n \models \varphi$ holds.
\end{theorem}

Soundness of certifying synthesis follows from the fact that every system process is required to satisfy its own certificate. Completeness is obtained since every strategy can serve as its own certificate: Intuitively, if $s_1 \pc \dots \pc s_n \models \varphi$, then LTL certificates that capture the exact behavior of the corresponding strategy satisfy the requirements of certifying synthesis. The proof is given in~\Cref{app:certifying_synthesis}.

Thus, certifying synthesis enables local~reasoning and therefore better understandability of the solution as well as modularity of the system, while ensuring to find correct solutions for all specifications that are realizable in the architecture.
Furthermore, the parallel composition of the strategies obtained with certifying synthesis for a specification $\varphi$ is a solution for the whole system.

%% file: computing_guarantees.tex
\section{Certifying Synthesis with Deterministic Certificates}\label{sec:assumption_lts}

There are several quality measures for certificates, for instance their size.
We, however, focus on certificates that are \emph{easy to synthesize}:
To determine whether a strategy sticks to its own certificate, a check for language containment has to be performed. Yet, efficient algorithms only exist for deterministic properties~\cite{TouatiBK95}. While certificates represented by LTL formulas are easily human-readable, they can be nondeterministic. Thus, the $\omega$-automaton representing the LTL certificate needs to be de\-ter\-minized, yielding an exponential blowup in its size~\cite{Safra88}.

In this section, we introduce a representation of certificates that ensures determinism to avoid the blowup.
Note that while enforcing determinism might yield larger certificates, it does not rule out any strategy that can be found with nondeterministic certificates: 
Since strategies are per se deterministic, there exists at least one deterministic certificate for them: The strategy itself.

We model the guaranteed behavior $g_i$ of a system process~$p_i$ as a labeled transition system~$\guarTrans{i}$, called \emph{guarantee transition system}~(GTS), over inputs~$\inputs{i}$ and \emph{guarantee output variables}~$\guarOutputs{i}\subseteq\outputs{i}$.
Only considering a subset of $\outputs{i}$ as output variables allows the certificate to abstract from outputs of $p_i$ whose valuation is irrelevant for all other system processes. In the following, we assume the guarantee output variables of $p_i$ to be both an output of $p_i$ and an input of some other system process, \ie, $\guarOutputs{i} := \outputs{i} \cap \allInputs$.
Intuitively, a variable $v \in \outputs{i} \setminus \guarOutputs{i}$ cannot be observed by any other process. Thus, a guarantee on its behavior does not influence any other system process and hence it can be omitted.
The variables~$\guarVariables{i}$ of the GTS of $p_i \in \sysProc$ are then given by $\guarVariables{i} := \inputs{i} \cup \guarOutputs{i}$.

In certifying synthesis, it is essential that a strategy only needs to satisfy the specification if the other processes do not deviate from their certificates. In the previous section, we used an implication in the local objective to model this.
When representing certificates as transition systems, we use \emph{valid histories} to determine whether a sequence matches the certificates of the other processes.

\begin{definition}[Valid History]
	Let~$\mathcal{G}_i$ be a set of guarantee transition systems.	
	A \emph{valid history of length~$t$ with respect to~$\mathcal{G}_i$} is a finite sequence $\sigma \in (2^V)^*$ of length~$t$, where for all $g_j \in \mathcal{G}_i$, $\sigma_k \cap \guarOutputs{j} = \comp(g_j,\hat{\sigma}\cap\inputs{j})_k \cap \guarOutputs{j}$ holds for all points in time $k$ with $1 \leq k \leq t$ and all infinite extensions $\hat{\sigma}$ of $\sigma$.
	The set of all valid histories of length $t$ with respect to $\mathcal{G}_i$ is denoted by $\validHistory{t}{\mathcal{G}_i}$.
\end{definition}

Intuitively, a valid history respecting a set $\mathcal{G}_i$ of guarantee transition systems is a finite sequence that is a prefix of a computation of all GTS in $\mathcal{G}_i$.
Thus, a valid history can be produced by the parallel composition of the GTS. 
Note that since strategies cannot look into the future, a finite word satisfies the requirements of a valid history either for all of its infinite extensions or for none of them.

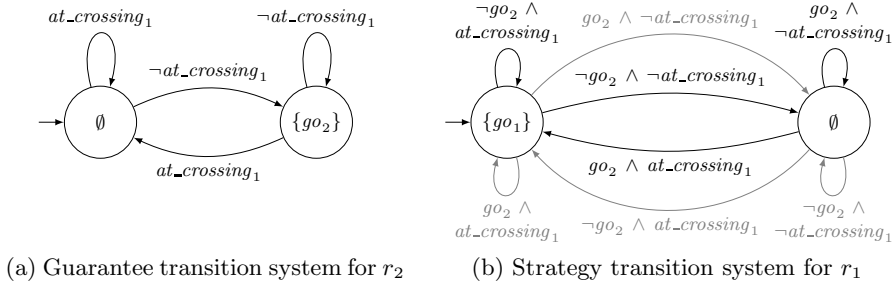
\begin{figure}[t]
	\begin{subfigure}{0.435\textwidth}
		\centering
		\scalebox{0.82}{
		\begin{tikzpicture}[>=latex,shorten >=0pt,auto,->,node distance=1cm,thin,every edge/.style={draw,font=\small}, initial text = ]
		
			\path[use as bounding box] (-1,-1.9) rectangle (4.5,1.85);
		
			\node[state,initial,minimum size=1.17cm]		(u0)		at (0,0)		{$\emptyset$};
			\node[state,minimum size=1.17cm]				(u1)		at (3.5,0)		{$\{\mathit{go}_2\}$};
			
			\path	(u0)	edge[bend left=25]	node		{$\neg \mathit{at\_crossing}_1$}	(u1)
						edge[loop above,looseness=10]		node		{$\mathit{at\_crossing}_1$}			(u0)
					(u1)	edge[bend left=25]	node		{$\mathit{at\_crossing}_1$}			(u0)
						edge[loop above,looseness=10]		node		{$\neg\mathit{at\_crossing}_1$}			(u1);
		\end{tikzpicture}}
		\caption{Guarantee transition system for $r_2$\label{fig:guarantee_r2}}
	\end{subfigure}
	\hfill
	\begin{subfigure}{0.56\textwidth}
		\centering
		\scalebox{0.82}{
		\begin{tikzpicture}[>=latex,shorten >=0pt,auto,->,node distance=1cm,thin,every edge/.style={draw,font=\small}, initial text = ]
		
			\path[use as bounding box] (-1,-1.9) rectangle (6.3,1.85);
		
			\node[state,initial,minimum size=1.17cm]		(t0)		at (0,0)		{$\{\mathit{go}_1\}$};
			\node[state,minimum size=1.17cm]				(t1)		at (5.3,0)		{$\emptyset$};
			
			\path	(t0)	edge[bend left=15]				node[align=center]		{$\neg \mathit{go}_2\,\land\,\neg \mathit{at\_crossing}_1$}	(t1)
						edge[loop above,looseness=7]	node[align=center]	{$\neg \mathit{go}_2~\land$ \\ $\mathit{at\_crossing}_1$}			(t0)
						edge[loop below,looseness=7,gray]	node[align=center]	{$\mathit{go}_2~\land$ \\ $\mathit{at\_crossing}_1$}	(t0)
						edge[bend left=48,gray]			node[align=center]	{$\mathit{go}_2\,\land\,\neg \mathit{at\_crossing}_1$}	(t1)
					(t1)	edge[bend left=15]				node	[align=center]	{$\mathit{go}_2\,\land\,\mathit{at\_crossing}_1$}			(t0)
						edge[bend left=48,gray]			node[align=center]	{$\neg\mathit{go}_2\,\land\,\mathit{at\_crossing}_1$}	(t0)
						edge[loop above,looseness=7]	node[align=center]	{$\mathit{go}_2~\land$ \\ $\neg\mathit{at\_crossing}_1$}			(t1)
						edge[loop below,looseness=7,gray]	node[align=center]	{$\neg\mathit{go}_2~\land$ \\ $\neg\mathit{at\_crossing}_1$}	(t1);
		\end{tikzpicture}}
		\caption{Strategy transition system for $r_1$\label{fig:strategy_r1}}
	\end{subfigure}
\caption{Strategy and GTS for robots $r_1$ and $r_2$ from \Cref{sec:motivating_example}, respectively. The labels of the states denote the output of the TS in the respective state.}\label{fig:example_lts}
\end{figure}

As an example for valid histories, consider the manufacturing robots again. Assume that $r_2$ guarantees to always give priority to $r_1$ at crossings and to move forward if $r_1$ is not at the crossing. A GTS $g_2$ for $r_2$ is depicted in \Cref{fig:guarantee_r2}.
Since $r_2$ never outputs~$\mathit{go}_2$ if~$r_1$ is at the crossing (left state), the finite sequence $\{ \mathit{at\_crossing}_1 \} \{ \mathit{go}_2 \}$ is no valid history respecting $g_2$. Since $r_2$ outputs~$\mathit{go}_2$ otherwise (right state), \eg, $\{\mathit{at\_crossing}_2 \} \{ \mathit{go}_2 \}$ is a valid history respecting $g_2$.

We use valid histories to determine whether the other processes stick to their certificates. Thus, intuitively, a strategy is required to satisfy the specification if its computation is a valid history respecting the GTS of the other processes:

\begin{definition}[Local Satisfaction]
	Let $\mathcal{G}_i$ be a set of guarantee transition systems.
	A strategy $s_i$ for $p_i\in\sysProc$ \emph{locally satisfies} an LTL formula $\varphi_i$ \emph{with respect to}~$\mathcal{G}_i$, denoted $s_i \satresp{\mathcal{G}_i} \varphi_i$, if $\comp(s_i,\gamma) \cup \gamma' \models \varphi_i$ holds for all $\gamma \in (2^\inputs{i})^\omega$, $\gamma' \in (2^{V \setminus \variables{i}})^\omega$ with $\pref{\comp(s_i,\gamma)}{t} \cup \pref{\gamma'}{t} \in \validHistory{t}{\mathcal{G}_i}$ for all points in time $t$.
\end{definition}

If $r_2$, for instance, sticks to its guaranteed behavior $g_2$ depicted in \Cref{fig:guarantee_r2}, then~$r_1$ can enter crossings regardless of $r_2$. Such a strategy $s_1$ for~$r_1$ is shown in \Cref{fig:strategy_r1}. Since neither $\sigma := \{ \mathit{at\_crossing}_1 \} \{ \mathit{go}_2 \}$ nor any finite sequence containing $\sigma$ is a valid history respecting $g_2$, no transition for input $\mathit{go}_2$ has to be considered for local satisfaction when $r_1$ is at the crossing (left state of $s_1$). 
Therefore, these transitions are depicted in gray. Analogously, no transition for $\neg\mathit{go}_2$ has to be considered when $r_1$ is not at the crossing (right state). 
The other transitions match valid histories and thus they are taken into account. Since no crash occurs when considering the black transitions only, $s_1 \satresp{\{g_2\}} \varphi_\mathit{safe}$ holds.

Using local satisfaction, we now define certifying synthesis in the setting where certificates are represented by labeled transition systems: Given an architecture $A$ and a specification $\varphi$, certifying synthesis for $\varphi$ derives strategies $s_1, \dots, s_n$ and \emph{guarantee transition systems} $g_1, \dots, g_n$ for the system processes.
For all $p_i \in \sysProc$, we require $s_i$ to locally satisfy its specification with respect to the guarantee transition systems of the other processes, \ie, $s_i \satresp{\mathcal{G}_i} \varphi_i$, where $\mathcal{G}_i = \{ g_j \mid p_j\in\sysProc\setminus\{p_i\}\}$. To ensure that a strategy does not deviate from its own certificate, $g_i$ is required to simulate $s_i$, \ie, $s_i \simresp{\guarOutputs{i}} g_i$ needs to hold.

In the following, we show that solutions of certifying synthesis with LTL cer\-ti\-fi\-cates can be translated into solutions with GTS and vice versa.
Given a solution of certifying synthesis with GTS, the main idea is to construct LTL certificates that capture the exact behavior of the GTS. 
For the formal certificate translation and its proof of correctness, we refer to~\Cref{app:computing_guarantees}.

\begin{lemma}\label{lem:guar_implies_cert}
	Let $\varphi$ be an LTL formula. 
	Let $\mathcal{S}$ and $\mathcal{G}$ be vectors of strategies and guarantee transition systems, respectively, for the system processes.
	If $(\mathcal{S},\mathcal{G})$ is a solution of certifying synthesis for~$\varphi$, then there exists a vector $\Psi$ of LTL~certificates such that $(\mathcal{S},\Psi)$ is a solution for certifying synthesis for $\varphi$ as well.
\end{lemma}

Given a solution of certifying synthesis with LTL certificates, we can construct GTS that match the strategies of the given solution. Then, these strategies as well as the GTS form a solution of certifying synthesis with GTS. The full construction and its proof of correctness is given in~\Cref{app:computing_guarantees}.

\begin{lemma}\label{lem:cert_implies_guar}
	Let $\varphi$ be an LTL formula. 
	Let $\mathcal{S}$ and $\Psi$ be vectors of strategies and LTL certificates, respectively, for the system processes.
	If $(\mathcal{S},\Psi)$ is a solution of certifying synthesis for $\varphi$, then there exists a vector $\mathcal{G}$ of guarantee transition system such that $(\mathcal{S},\mathcal{G})$ is a solution for certifying synthesis for $\varphi$ as well.
\end{lemma}

Hence, we can translate solutions of certifying synthesis with LTL formulas and with GTS into each other. Thus, we can reuse the results from \Cref{sec:certifying_synthesis}, in particular \Cref{thm:soundness_completeness_certifying_synthesis}, and then soundness and completeness of certifying synthesis with guarantee transition systems follows with \Cref{lem:guar_implies_cert,lem:cert_implies_guar}:

\begin{theorem}[Soundness and Completeness with GTS]\label{thm:correctness_assumption_lts}
	Let~$\varphi$ be an LTL formula.
	Let $\mathcal{S} = \myVec{s_1, \dots, s_n}$ be a vector of strategies for the system processes.
	Then, there exists a vector $\mathcal{G}$ of guarantee transition systems such that $(\mathcal{S},\mathcal{G})$ is a solution of certifying synthesis for $\varphi$ if, and only if, $s_1 \pc \dots \pc s_n \models \varphi$ holds.
\end{theorem}

Thus, similar to LTL certificates, certifying synthesis with GTS allows for local reasoning and thus enables modularity of the system while it still ensures that correct solutions for all realizable specifications are found. In particular, enforcing deterministic certificates does not rule out strategies that can be obtained with either nondeterministic certificates or with classical distributed synthesis.

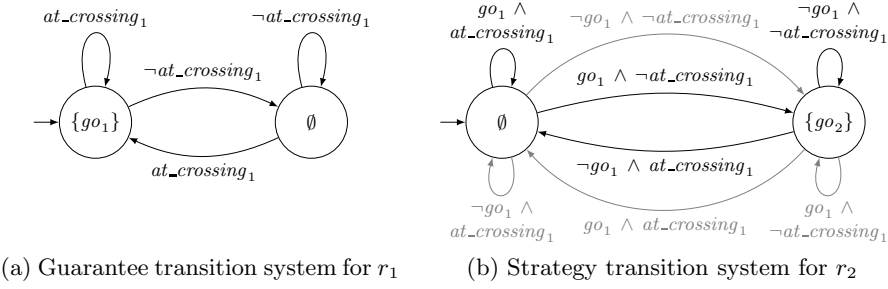
\begin{figure}[t]
	\begin{subfigure}{0.435\textwidth}
		\centering
		\scalebox{0.82}{
		\begin{tikzpicture}[>=latex,shorten >=0pt,auto,->,node distance=1cm,thin,every edge/.style={draw,font=\small}, initial text = ]
		
			\path[use as bounding box] (-1,-1.9) rectangle (4.5,1.85);
		
			\node[state,initial,minimum size=1.17cm]		(u0)		at (0,0)		{$\{\mathit{go}_1\}$};
			\node[state,minimum size=1.17cm]				(u1)		at (3.5,0)		{$\emptyset$};
			
			\path	(u0)	edge[bend left=25]	node		{$\neg \mathit{at\_crossing}_1$}	(u1)
						edge[loop above,looseness=10]		node		{$\mathit{at\_crossing}_1$}			(u0)
					(u1)	edge[bend left=25]	node		{$\mathit{at\_crossing}_1$}			(u0)
						edge[loop above,looseness=10]		node		{$\neg \mathit{at\_crossing}_1$}	(u1);
		\end{tikzpicture}}
		\caption{Guarantee transition system for $r_1$\label{fig:guarantee_r1}}
	\end{subfigure}
	\hfill
	\begin{subfigure}{0.56\textwidth}
		\centering
		\scalebox{0.82}{
		\begin{tikzpicture}[>=latex,shorten >=0pt,auto,->,node distance=1cm,thin,every edge/.style={draw,font=\small}, initial text = ]
		
			\path[use as bounding box] (-1,-1.9) rectangle (6.3,1.85);
		
			\node[state,initial,minimum size=1.17cm]		(t0)		at (0,0)		{$\emptyset$};
			\node[state,minimum size=1.17cm]				(t1)		at (5.3,0)		{$\{\mathit{go}_2\}$};
			
			\path	(t0)	edge[bend left=15]				node[align=center]		{$\mathit{go}_1\,\land\,\neg \mathit{at\_crossing}_1$}	(t1)
						edge[loop above,looseness=7]	node[align=center]	{$\mathit{go}_1~\land$ \\ $\mathit{at\_crossing}_1$}			(t0)
						edge[loop below,looseness=7,gray]	node[align=center]	{$\neg \mathit{go}_1~\land$ \\ $\mathit{at\_crossing}_1$}	(t0)
						edge[bend left=48,gray]			node[align=center]	{$\neg \mathit{go}_1\,\land\,\neg \mathit{at\_crossing}_1$}	(t1)
					(t1)	edge[bend left=15]				node	[align=center]	{$\neg\mathit{go}_1\,\land\,\mathit{at\_crossing}_1$}			(t0)
						edge[loop above,looseness=7]	node[align=center]		{$\neg\mathit{go}_1~\land$ \\ $\neg \mathit{at\_crossing}_1$}	(t1)
						edge[loop below,looseness=7,gray]	node[align=center]	{$\mathit{go}_1~\land$ \\ $\neg \mathit{at\_crossing}_1$}	(t1)
						edge[bend left=48,gray]			node[align=center]	{$\mathit{go}_1\,\land\,\mathit{at\_crossing}_1$}	(t0);
		\end{tikzpicture}}
		\caption{Strategy transition system for $r_2$\label{fig:strategy_r2}}
	\end{subfigure}
\caption{GTS and strategy for robots $r_1$ and $r_2$ from \Cref{sec:motivating_example}, respectively. The labels of the states denote the output of the TS in the respective state.}\label{fig:example_lts_2}
\end{figure}

As an example of the whole synthesis procedure of a distributed system with certifying synthesis and GTS, consider the manufacturing robots from \Cref{sec:motivating_example}. For simplicity, suppose that the robots do not have individual additional requirements $\varphi_{\mathit{add}_i}$. Hence, the full specification is given by $\varphi_\mathit{safe} \land \varphi_{\mathit{cross}_1} \land \varphi_{\mathit{cross}_2}$. Since $\mathit{go}_i$ is an output variable of robot $r_i$, we obtain the subspecifications $\varphi_i = \varphi_\mathit{safe} \land \varphi_{\mathit{cross}_i}$. A solution of certifying synthesis is then given by the strategies and GTS depicted in \Cref{fig:example_lts,fig:example_lts_2}. 
Note that $s_2$ only locally satisfies $\varphi_{\mathit{cross}_2}$ with respect to $g_1$ when assuming that $r_1$ is not immediately again at the intersection after crossing it. However, there are solutions with slightly more complicated certificates that do not need this assumption.
The parallel composition of $s_1$ and $s_2$ yields a strategy that allows $r_1$ to move forwards if it is at the crossing and that allows $r_2$ to move forwards otherwise.

%% file: relevant_processes.tex
\section{Computing Relevant Processes}

Both representations of certificates introduced in the last two sections consider the certificates of \emph{all} other system processes in the local objective of every system process $p_i$.
This is not always necessary since in some cases $\varphi_i$ is satisfiable by a strategy for $p_i$ even if another process deviates from its guaranteed behavior.

In this section, we present an optimization of certifying synthesis that reduces the number of considered certificates.
We compute a set of \emph{relevant processes $\relevantProcesses{i} \subseteq \sysProc\setminus\{p_i\}$} for every $p_i \in \sysProc$. 
Certifying synthesis then only considers the certificates of the relevant processes:
For LTL certificates, it requires that $s_i \models \gb{i} \land (\relGB{i} \rightarrow \varphi_i)$ holds, where $\relGB{i} = \{ \gb{j} \in \Psi \mid p_j \in \relevantProcesses{i} \}$.
For GTS, both $s_i \simresp{\guarOutputs{i}} g_i$ and $s_i \satresp{\mathcal{G}^\mathcal{R}_i} \varphi_i$ need to hold, where $\mathcal{G}^\mathcal{R}_i = \{ g_j \in \mathcal{G} \mid p_j \in \relevantProcesses{i}\}$.
Such solutions of certifying synthesis are denoted by $(\mathcal{S},\Psi)_\mathcal{R}$ and $(\mathcal{S},\mathcal{G})_\mathcal{R}$, respectively.

The construction of the relevant processes $\relevantProcesses{i}$ has to ensure that certifying synthesis is still sound and complete. In the following, we introduce a definition of relevant processes that does so.
It excludes processes from $p_i$'s set of relevant processes $\relevantProcesses{i}$ whose output variables do not occur in the subspecification~$\varphi_i$:

\begin{definition}[Relevant Processes]\label{def:ref_rel_proc}
	Let $\varphi$ be an LTL formula with decomposition $\myVec{\varphi_1,\dots,\varphi_n}$. The \emph{relevant processes $\relevantProcesses{i} \subseteq \sysProc\setminus\{p_i\}$} of system process $p_i \in \sysProc$ are given by $\relevantProcesses{i} = \{ p_j \in \sysProc \setminus \{p_i\} \mid \outputs{j} \cap \propositions{\varphi_i} \neq \emptyset \}$.
\end{definition}

Intuitively, since $\outputs{j} \cap \propositions{\varphi_i} = \emptyset$ holds for a process $p_j \in \sysProc \setminus \relevantProcesses{i}$ with $i \neq j$, the subspecification $\varphi_i$ does not restrict the satisfying valuations of the output variables of $p_j$. Thus, in particular, if a sequence satisfies $\varphi_i$, then it does so for any valuations of the variables in $\outputs{j}$. Hence, the guaranteed behavior of $p_j$ does not influence the satisfiability of $\varphi_i$ and thus $p_i$ does not need to consider it.
The proof of the following theorem stating this property is given in \Cref{app:relevant_processes}.

\begin{theorem}[Correctness of Relevant Processes]\label{thm:correctness_rel_proc}
	Let $\varphi$ be an LTL formula. Let $\mathcal{S} = \myVec{s_1, \dots, s_n}$ be a vector of strategies for the system processes.
	\begin{enumerate}
		\item Let $\Psi$ be a vector of LTL certificates. If $(\mathcal{S},\Psi)_{\mathcal{R}}$ is a solution of certifying synthesis for $\varphi$, then $s_1 \pc \dots \pc s_n \models \varphi$ holds.  If $s_1 \pc \dots \pc s_n \models \varphi$ holds, then there exists a vector $\Psi'$ of LTL certificates and a vector~$\mathcal{S}'$ of strategies such that $(\mathcal{S}',\Psi')_{\mathcal{R}}$ is a solution of certifying synthesis for $\varphi$.
		\item Let $\mathcal{G}$ be a vector of guarantee transition systems. If $(\mathcal{S},\mathcal{G})_{\mathcal{R}}$ is a solution of certifying synthesis for $\varphi$, then $s_1 \pc \dots \pc s_n \models \varphi$. If $s_1 \pc \dots \pc s_n \models \varphi$ holds, then there exists a vector $\mathcal{G}'$ of guarantee transition systems and a vector $\mathcal{S}'$ of strategies such that $(\mathcal{S}',\mathcal{G}')_{\mathcal{R}}$ is a solution of certifying synthesis for $\varphi$.
	\end{enumerate}
\end{theorem}

Note that for certifying synthesis with relevant processes, we can only guarantee that for every vector of strategies $\myVec{s_1,\dots,s_n}$ whose parallel composition satisfies the specification, there exist \emph{some} strategies that are a solution of certifying synthesis. These strategies are not necessarily $s_1, \dots, s_n$: A strategy $s_i$ may make use of the certificate of a process $p_j$ outside of $\relevantProcesses{i}$. That is, it may violate its specification~$\varphi_i$ on an input sequence that does not stick to $g_j$ although $\varphi_i$ is satisfiable for this input. Strategy $s_i$ is not required to satisfy $\varphi_i$ on this input, a strategy that may only consider the certificates of the relevant processes, however, is.
As long as the definition of relevant processes allows for finding \emph{some} solution of certifying synthesis, like the one introduced in this section does as a result of \Cref{thm:correctness_rel_proc}, certifying synthesis is nevertheless sound and complete.

%% file: synthesizing_certificates.tex
\section{Synthesizing Certificates}

In this section, we describe an algorithm for practically synthesizing strategies and deterministic certificates represented by GTS.
Our approach is based on \emph{bounded synthesis}~\cite{FinkbeinerS13} and bounds the size of the strategies and of the certificates. This allows for producing size-optimal solutions in either terms of strategies or certificates.
Like for monolithic bounded synthesis~\cite{FinkbeinerS13,FaymonvilleFRT17}, we encode the search for a solution of certifying synthesis of a certain size into a SAT constraint system. We reuse parts of the constraint system for monolithic systems.

An essential part of bounded synthesis is to determine whether a strategy satisfies an LTL formula $\varphi_i$. To do so, we first construct the equivalent universal co-Büchi automaton $\mathcal{A}_i$ with $\mathcal{L}(\mathcal{A}_i) = \mathcal{L}(\varphi_i)$. Then, we check whether~$\mathcal{A}_i$ accepts $\comp(s_i,\gamma)\cup\gamma'$ for all $\gamma \in (2^\inputs{i})^\omega$, $\gamma' \in (2^{V\setminus\variables{i}})^\omega$, \ie, whether all runs of $\mathcal{A}_i$ induced by $\comp(s_i,\gamma)\cup\gamma'$ contain only finitely many visits to rejecting states.
So far, we used \emph{local satisfaction} to formalize that in compositional synthesis with GTS a strategy only needs to satisfy its specification as long as the other processes stick to their guarantees. That is, we changed the satisfaction condition.
To reuse existing algorithms for bounded synthesis and, in particular, for checking whether a strategy is winning, however, we incorporate this property of certifying synthesis into the labeled transition system representing the strategy instead. In fact, we utilize the following observation:
A finite run of a universal co-Büchi automaton can never visit a rejecting state infinitely often.
 Hence, by ensuring that the automaton produces finite runs on all sequences that deviate from a guarantee, checking whether a strategy satisfies a specification can still be done by checking whether the runs of the corresponding automaton induced by the computations of the strategy visit a rejecting state only finitely often.

Therefore, we represent strategies by \emph{incomplete} transition systems in the following. The domain of definition of their transition function is defined such that the computation of a strategy is infinite if, and only if, the other processes stick to their guarantees. To formalize this, we utilize valid histories:

\begin{definition}[Local Strategy]
	A \emph{local strategy} $s_i$ for process $p_i \in \sysProc$ \emph{with respect to} a set $\mathcal{G}_i$ of GTS is represented by a TS $\mathcal{T}_i = (T,t_0,\tau,o)$ with a partial transition function $\tau: T \times 2^\inputs{i} \rightharpoonup T$.
	The domain of definition of $\tau$ is defined such that $\comp(s_i,\gamma)$ is infinite for $\gamma \in (2^\inputs{i})^\omega$ if, and only if, there exists $\gamma' \in (2^{V \setminus \variables{i}})^\omega$ such that $\pref{\comp(s_i,\gamma)}{t} \cup \pref{\gamma'}{t} \in\validHistory{t}{\mathcal{G}_i}$ holds for all points in time $t$.
\end{definition}

As an example, consider strategy $s_1$ for robot $r_1$ and guarantee transition system~$g_2$ for robot $r_2$, both depicted in \Cref{fig:example_lts}, again. From $s_1$, we can construct a local strategy $s'_1$ for $r_1$ with respect to $g_2$ by eliminating the gray transitions.

We now define certifying synthesis with local strategies: Given a specification~$\varphi$, certifying synthesis derives GTS $g_1,\dots,g_n$ and \emph{local strategies} $s_1,\dots,s_n$ respecting these guarantees, such that for all $p_i \in \sysProc$, $s_i \simresp{\guarOutputs{i}} g_i$ holds and all runs of $\mathcal{A}_i$ induced by $\comp(s_i,\gamma)\cup\gamma'$ contain finitely many visits to re\-jec\-ting states for all $\gamma \in (2^\inputs{i})^\omega$, $\gamma' \in (2^{V \setminus \variables{i}})^\omega$, where $\mathcal{A}_i$ is a universal co-Büchi automaton with $\mathcal{L}(\mathcal{A}_i)=\mathcal{L}(\varphi_i)$. 
Thus, we can reuse existing algorithms for checking satisfaction of a formula in our certifying synthesis algorithm when synthesizing local strategies instead of complete ones. Similar to monolithic bounded synthesis, we construct a constraint system encoding the search for local strategies and~GTS:

\begin{theorem}\label{thm:constraint_system}
	Let $A$ be an architecture, let $\varphi$ be an LTL formula, and let $\mathcal{B}$ be the size bounds. There is a SAT constraint system~$\constraintSystem{A}{\varphi}{\mathcal{B}}$ such that (1) if $\constraintSystem{A}{\varphi}{\mathcal{B}}$ is satisfiable, then~$\varphi$ is realizable in $A$, (2) if~$\varphi$ is realizable in~$A$ for the bounds~$\mathcal{B}$ and additionally $\propositions{\varphi_i} \subseteq \variables{i}$ holds for all $p_i\in\sysProc$, then $\constraintSystem{A}{\varphi}{\mathcal{B}}$ is satisfiable.
\end{theorem}

Intuitively, the constraint system $\constraintSystem{A}{\varphi}{\mathcal{B}}$ consists of $n$ slightly adapted copies of the SAT constraint system for monolithic systems~\cite{FinkbeinerS13,FaymonvilleFRT17} as well as additional constraints that ensure that the synthesized local strategies correspond to the synthesized guarantees and that they indeed fulfill the conditions of certifying synthesis. The constraint system $\constraintSystem{A}{\varphi}{\mathcal{B}}$ is presented in~\Cref{app:synthesizing_certificates}.

Note that we build a \emph{single} constraint system for the whole certifying synthesis task. That is, the strategies and certificates of the individual processes are not synthesized completely independently. This is one of the main differences of our approach to the negotiation-based assume-guarantee synthesis algorithm~\cite{MajumdarMSZ20}.
While this prevents separate synthesis tasks and thus parallelizability, it eliminates the need for a negotiation between the processes. Moreover, it allows for completeness of the synthesis algorithm.
Although the synthesis tasks are not fully separated, the constraint system $\constraintSystem{A}{\varphi}{\mathcal{B}}$ is in most cases still significantly smaller and easier to solve than the one of classical distributed synthesis.

As indicated in \Cref{thm:constraint_system}, certifying synthesis with local strategies is not complete in general: We can only ensure completeness if the satisfaction of each subspecification solely depends on the variables that the corresponding process can observe.
This incompleteness is due to a slight difference in the satisfaction of a specification with local strategies and local satisfaction with complete strategies: The latter requires a strategy~$s_i$ to satisfy $\varphi_i$ if \emph{all processes} stick to their guarantees.
The former, in contrast, requires $s_i$ to satisfy~$\varphi_i$ if \emph{all processes producing observable outputs} stick to their guarantees.
Hence, if~$p_i$ cannot observe whether $p_j$ sticks to its guarantee, satisfaction with local strategies requires~$s_i$ to satisfy $\varphi_i$ even if $p_j$ deviates, while local satisfaction does not.

This slight change in definition is needed in order to incorporate the requirements of certifying synthesis into the transition system representing the strategy and thus to be able to reuse existing bounded synthesis frameworks. 
Although this advantage is at general completenesses expanse, we experienced that in practice many distributed systems, at least after rewriting the specification, indeed satisfy the condition that is needed for completeness in our approximation of certifying synthesis.
In fact, all benchmarks described in \Cref{sec:experiments} satisfy it.

%% file: experiments.tex
\section{Experimental Results}\label{sec:experiments}

We have implemented certifying synthesis with local strategies and guarantee transition systems. It expects an LTL formula and its decomposition as well as the system architecture, and bounds on the sizes of the strategies and certificates as input.
Specification decomposition can easily be automated by, \eg, implementing \Cref{def:spec_decomposition}.
The implementation extends the synthesis tool BoSy~\cite{FaymonvilleFT17} for monolithic systems to certifying synthesis for distributed systems. In particular, we extend and adapt BoSy's SAT encoding~\cite{FaymonvilleFRT17} as described in~\Cref{app:synthesizing_certificates}.

\begin{table}[t]
    \centering
    \caption{Experimental results on scalable benchmarks. Reported is the parameter and the running time in seconds. We used a machine with a 3.1 GHz Dual-Core Intel Core i5 processor and 16 GB of RAM, and a timeout of 60min. For dist.\ BoSy, we use the SMT encoding and give the average runtime of 10~runs.\\}\label{table:results}
    \begin{tabular}{p{3.2cm}>{\centering}p{1.3cm}||>{\centering}p{2.1cm}|>{\centering}p{2.1cm}|>{\centering\arraybackslash}p{2.1cm}}
	     Benchmark & Param. & Cert. Synth. & Dist. BoSy & Dom. Strat.\\
	     \hline\hline
	     n-ary Latch & 2 & \textbf{0.89} & 41.26 & 4.75\\
	     & 3 & \textbf{0.91} & TO & 6.40\\
		 & \dots & \dots & \dots & \dots \\
	     & 6 & \textbf{12.26} & TO & 13.89\\
	     & 7 & 105.69 & TO & \textbf{15.06}\\
	     \hline
	     Generalized Buffer & 1 & \textbf{1.20} & 6.59 & 5.23\\
	     & 2 & \textbf{2.72} & 3012.51 & 10.53\\
	     & 3 & \textbf{122.09} & TO & 961.60\\
	     \hline
	     Load Balancer & 1 & \textbf{0.98} & 1.89 & 2.18\\
	     & 2 & \textbf{1.64} & 2.39 & --\\
	     \hline
	     Shift & 2 & \textbf{1.10} & 1.99 & 4.76 \\
	     & 3 & \textbf{1.13} & 4.16 & 7.04 \\
	     & 4 & \textbf{1.14} & TO & 11.13 \\
		 & \dots & \dots & \dots & \dots \\
	     & 7 & \textbf{9.01} & TO & 16.08 \\
	     & 8 & 71.89 & TO & \textbf{19.38} \\
	     \hline
	     Ripple-Carry Adder & 1 & \textbf{0.878} & 1.83 & --\\
	     & 2 & \textbf{2.09} & 36.84 & -- \\
	     & 3 & \textbf{106.45} & TO & -- \\
	     \hline
	     Manufacturing Robots & 2 & \textbf{1.10} & 2.45 & --\\
	     & 4 & \textbf{1.18} & 2.43 & --\\
	     & 6 & \textbf{1.67} & 3.20 & --\\
	     & 8 & \textbf{2.88} & 5.67 & --\\
	     & 10 & \textbf{48.83} & 221.16 & --\\
	     & 12 & \textbf{1.44} & TO & --\\
		 & \dots & \dots & \dots & \dots \\
	     & 42 & \textbf{373.90} & TO & --\\
    \end{tabular}
\end{table}

We compare our implementation to two extensions of BoSy: One for distributed systems~\cite{Baumeister17} and one for synthesizing individual dominant strategies, implementing the compositional synthesis algorithm presented in~\cite{DammF14}.
The results are shown in \Cref{table:results}. 
We used the SMT encoding of distributed~BoSy since the other ones either cause memory errors on almost all benchmarks~(SAT), or do not support most of our architectures (QBF).
Since the running times of the underlying SMT solver vary immensely, we report on the average running time of 10 runs.
Synthesizing individual dominant strategies is incomplete and hence we can only report on results for half of our benchmarks.
We could not compare our implementation to the iterative assume-guarantee synthesis tool Agnes~\cite{MajumdarMSZ20}, since it currently does not support most of our architectures or specifications.

The first four benchmarks stem from the synthesis competition~\cite{SYNTCOMP2018}. The latch is parameterized in the number of bits, the generalized buffer in the number of senders, the load balancer in the number of servers, and the shift benchmark in the number of inputs.
The fourth benchmark is a ripple-carry adder that is parameterized in the number of bits and the last benchmark describes the manufacturing robots from \Cref{sec:motivating_example} and is parameterized in the size of the objectives $\varphi_{\mathit{add}_i}$ of the robots.
The system architectures are given in~\Cref{app:benchmarks}.

For the latch, the generalized buffer, the ripple-carry adder, and the shift, certifying synthesis clearly outperforms distributed BoSy. 
For many parameters, BoSy does not terminate within 60min, while certifying synthesis solves the tasks in less than 13s.
For these benchmarks, a process does not~need to know the full behavior of the other processes. 
Hence, the certificates are notably smaller than the strategies. 
A process of the ripple-carry adder, for instance, only needs information about the carry bit of the previous process, the sum bit is irrelevant.

For the load balancer, in contrast, the certificates need to contain the full behavior of the processes.
Hence, the benefit of the compositional approach lies solely in the specification decomposition.
This advantage suffices to produce a solution faster than distributed BoSy.
Yet, for other benchmarks with full certificates, the overhead of synthesizing certificates dominates the benefit of specification decomposition for larger parameters, showcasing that certifying synthesis is particularly beneficial if a small interface between the processes exists. 

The manufacturing robot benchmark is designed such that the interface between the processes stays small for all parameters. Hence, it demonstrates the advantage of abstracting from irrelevant behavior. Certifying synthesis clearly outperforms distributed BoSy on all instances. The parameter corresponds to the minimal solution size with distributed BoSy which does not directly correspond to the solution size with certifying synthesis. Thus, the running times do not grow in parallel. For more details on this benchmark, including a more detailed table on the results, we refer to~\Cref{app:benchmarks}.

Thus, certifying synthesis is extremely beneficial for specifications where small certificates exist. This directly corresponds to the existence of a small interface between the processes of the system. Hence, bounding the size of the certificates indeed guides the synthesis procedure in finding solutions fast.

When synthesizing dominant strategies, the weaker winning condition poses implicit assumptions on the behavior of the other processes.
These assumptions do not always suffice: There are no independent dominant strategies for the load balancer, the ripple-carry adder, and the robots.
For the other bench\-marks, the algorithm terminates. 
While certifying synthesis performs better for the generalized buffer, the slight overhead of synthesizing explicit certificates becomes clear for the latch and the shift: 
For small parameters, certifying synthesis produces a solution faster. 
For larger parameters, synthesizing dominant strategies outperforms certifying synthesis.
Yet, the implicit assumptions do not encapsulate the required interface between the processes and thus they do not increase the understandability of the system's interconnections.

%% file: conclusions.tex
\section{Conclusions}

We have presented a synthesis algorithm that reduces the complexity of distributed synthesis by decomposing the synthesis task into smaller ones for the individual processes.
To ensure completeness, the algorithm synthesizes additional certificates that capture a certain behavior a process commits to. A process then makes use of the certificates of the other processes by only requiring its strategy to satisfy the specification if the other processes do not deviate from their certificates.
Synthesizing additional certificates increases the understandability of the system and the solution since the certificates capture the interconnections of the processes and which agreements they have to establish.
Moreover, the certificates form a contract between the processes: The synthesized strategies can be substituted as long as the new strategy still complies with the contract, \ie, as long as it does not deviate from the guaranteed behavior, enabling modularity.

We have introduced two representations of the certificates, as LTL formulas and as labeled transition systems. Both ensure soundness and completeness of the compositional certifying synthesis algorithm. For the latter representation, we presented an encoding of the search for strategies and certificates into a SAT constraint solving problem.
Moreover, we have introduced a technique for reducing the number of certificates that a process needs to consider by determining relevant processes.
We have implemented the certifying synthesis algorithm and compared it to two extensions of the synthesis tool BoSy to distributed systems.
The results clearly show the advantage of compositional approaches as well as of guiding the synthesis procedure by bounding the size of the certificates: For benchmarks where small interfaces between the processes exist, certifying synthesis outperforms the other distributed synthesis tools significantly. If no solution with small interfaces exist, the overhead of certifying synthesis is small.

%% file: appendix.tex
\section{Compositional Synthesis with Certificates}\label{app:certifying_synthesis}

\subsection*{Soundness and Completeness of Certifying Synthesis (Proof of \Cref{thm:soundness_completeness_certifying_synthesis})}

First, we prove the soundness of certifying synthesis. Intuitively, it follows from the fact that every system process is required to satisfy its own certificate.

\begin{lemma}[Soundness of Certifying Synthesis]\label{lem:soundness_certifying_synthesis}
	Let $\varphi$ be an LTL formula. Let $\mathcal{S} = \myVec{s_1, \dots, s_n}$ and $\Psi = \myVec{\gb{1}, \dots, \gb{n}}$ be vectors of strategies and LTL certificates, respectively, for the system processes.
	If $(\mathcal{S},\Psi)$ is a solution of certifying synthesis for $\varphi$, then $s_1 \pc \dots \pc s_n \models \varphi$ holds as well.
\end{lemma}
\begin{proof}
	Let $(\mathcal{S},\Psi)$ be a solution of certifying synthesis for $\varphi$.
	Furthermore, let $\GB{i} = \{ \gb{j} \mid p_j \in \sysProc \setminus \{p_i\} \}$.
	Then, by definition, $s_i \models \gb{i} \land (\GB{i} \rightarrow \varphi_i)$ holds for all system processes $p_i \in \sysProc$. 
	Let $\gamma \in (2^{\envOutputs})^\omega$ and, for the sake of better readability, let $\sigma = \comp(s_1 \pc\dots\pc s_n,\gamma)$. Then, in particular, \[\comp(s_i,\sigma\cap\inputs{i})\cup(\sigma\cap(V\setminus\variables{i})) \models \gb{i} \land (\GB{i} \rightarrow \varphi_i)\] holds for all $p_i \in \sysProc$.
	By the definition of parallel composition and since the sets of output variables are pairwise disjoint, we have \[\comp(s_i,\sigma\cap\inputs{i})\cup(\sigma\cap(V\setminus\variables{i})) = \sigma\] for all $p_i \in \sysProc$. Thus, $\sigma \models \gb{i} \land (\GB{i} \rightarrow \varphi_i)$ follows for all $p_i \in \sysProc$ and therefore
	\[\comp(s_1 \pc \dots \pc s_n, \gamma) \models (\gb{1} \land (\GB{1} \rightarrow \varphi_1)) \land \dots \land (\gb{n} \land (\GB{n} \rightarrow \varphi_n))\] follows.
	Thus, in particular $\comp(s_1 \pc \dots \pc s_n, \gamma) \models \gb{1} \land \dots \land \gb{n}$ holds and hence, by definition of $\GB{i}$, $\comp(s_1 \pc \dots \pc s_n,\gamma) \models \GB{i}$ follows for all $i$ with $1 \leq i \leq n$.
	Therefore, by the semantics of implication, $\comp(s_1 \pc \dots \pc s_n, \gamma) \models \varphi_1 \land \dots \land \varphi_n$ holds as well.
	Hence, since we chose an arbitrary $\gamma \in (2^\envOutputs)^\omega$ above, we have $\comp(s_1 \pc \dots \pc s_n, \gamma) \models \varphi_1 \land \dots \land \varphi_n$ for all $\gamma \in (2^{\envOutputs})^\omega$ and thus $s_1 \pc \dots \pc s_n \models \varphi$ follows with the definition of specification decomposition.\qed

\end{proof}

Second, we prove the completeness of certifying synthesis, \ie, we show that if $s_1 \pc \dots \pc s_n \models \varphi$ holds, then there exist certificates $\varphi_1, \dots, \varphi_n$ such that the strategies and the certificates build a solution of certifying synthesis for $\varphi$.. Intuitively, LTL formulas that capture the exact behavior of the corresponding strategy satisfy the requirements of certifying synthesis.

\begin{lemma}[Completeness of Certifying Synthesis]\label{lem:completeness_certifying_synthesis}
	Let $\varphi$ be an LTL formula. Let $\mathcal{S} = \myVec{s_1, \dots, s_n}$ be a vector of strategies for the system processes.
	If $s_1 \pc \dots \pc s_n \models \varphi$ holds, then there exists a vector $\Psi$ of LTL certificates such that $(\mathcal{S},\Psi)$ is a solution of certifying synthesis for $\varphi$.
\end{lemma}
\begin{proof}
	Let $s_1 \pc \dots \pc s_n \models \varphi$. We construct LTL certificates $\gb{1}, \dots, \gb{n}$ as follows:~$\gb{i}$ describes exactly the behavior of $s_i$, \ie, every computation of $s_i$ satisfies~$\gb{i}$ and there is no trace that satisfies $\gb{i}$ but that is no computation of $s_i$. Hence, $\mathcal{L}(\gb{i}) = \{ \comp(s_i,\gamma) \cup \gamma' \mid \gamma \in (2^\inputs{i})^\omega, \gamma' \in (2^{V \setminus \variables{i}})^\omega \}$ holds. Since a strategy~$s_i$ is represented by a finite-state labeled transition system, the construction of such an LTL formula~$\psi_i$ is always possible. Let $\GB{i} = \{ \gb{j} \mid p_j \in \sysProc \setminus \{p_i\} \}$.
		
	Note that by construction, for all sequences $\sigma \in (2^V)^\omega$, we have $\sigma \models \gb{i}$ if, and only if, $\sigma \cap \variables{i} = \comp(s_i,\gamma') \cap \variables{i}$ holds for some $\gamma' \in (2^{\inputs{i}})^\omega$.
	Thus, $\sigma \models \gb{1} \land \dots \land \gb{n}$ holds for some $\sigma \in (2^V)^\omega$ if, and only if, for all $i$ with $1 \leq i \leq n$, there exists $\gamma' \in (2^{\inputs{i}})^\omega$ such that $\sigma \cap \variables{i} = \comp(s_i,\gamma') \cap \variables{i}$ holds.
	By disjointness of the inputs and outputs of a process, we have $\comp(s_i,\gamma') \cap \inputs{i} = \gamma' \cap \inputs{i}$ for all $1 \leq i \leq n$ and all $\gamma' \in (2^{\inputs{i}})^\omega$. Thus, $\sigma \cap \inputs{i} = \gamma'$ follows since $\inputs{i} \subseteq \variables{i}$. Therefore, $\sigma \cap \variables{i} = \comp(s_i,\sigma \cap \inputs{i}) \cap \variables{i}$ holds for all $1 \leq i \leq n$.
	Hence, by definition of the parallel composition of strategies, we have $\sigma = \comp(s_1 \pc \dots \pc s_n,\sigma \cap \envOutputs)$ for all $\sigma \in (2^V)^\omega$ with $\sigma \models \gb{1}\land \dots \land \gb{n}$ since $\allInputs \setminus \allOutputs = \envOutputs$.
	
	It remains to show that $s_i \models \gb{i} \land (\GB{i} \rightarrow \varphi_i)$ holds for all $i$ with $1 \leq i \leq n$.
	Let $p_i \in \sysProc$. By construction of the LTL formulas, we clearly have $s_i \models \gb{i}$. 
	Next, let $\gamma \in (2^{\inputs{i}})^\omega$, $\gamma' \in (2^{V\setminus\variables{i}})^\omega$. We distinguish two cases:
	\begin{enumerate}
		\item Let $\comp(s_i,\gamma) \cup \gamma' \models \neg\GB{i}$ hold. Then, $\comp(s_i,\gamma) \cup \gamma' \models \GB{i} \rightarrow \varphi_i$ follows directly with the semantics of implication. 
		\item Let $\comp(s_i,\gamma) \cup \gamma' \not\models \neg\GB{i}$ hold. Then, we have $\comp(s_i,\gamma) \cup \gamma' \models \GB{i}$. By construction of $\GB{i}$ and since, as shown above, $s_i \models \gb{i}$ holds, we thus have $\comp(s_i,\gamma) \cup \gamma' \models \gb{1} \land \dots \land \gb{n}$. Therefore, by construction of the $\gb{i}$, \[\comp(s_i,\gamma) \cup \gamma' = \comp(s_1 \pc \dots \pc s_n,(\comp(s_i,\gamma)\cup\gamma') \cap \envOutputs)\] follows as shown above. 
		Since the sets of output variables in an architecture are pairwise disjoint, in particular $\outputs{i} \cap \envOutputs = \emptyset$ holds. Hence, we have $\comp(s_i,\gamma)\cap\envOutputs = \gamma\cap\envOutputs$ and thus \[\comp(s_i,\gamma) \cup \gamma' = \comp(s_1 \pc \dots \pc s_n,(\gamma\cup\gamma') \cap \envOutputs)\] follows.
		Since we have $s_1 \pc \dots \pc s_n \models \varphi$ by assumption, $s_1 \pc \dots \pc s_n \models \varphi_i$ holds by the definition of specification decomposition and the semantics of conjunction as well. Hence, $\comp(s_1 \pc \dots \pc s_n,(\gamma\cup\gamma')\cap\envOutputs) \models \varphi_i$ holds and thus $\comp(s_i,\gamma) \cup \gamma' \models \varphi_i$ follows. Hence, by the semantics of implication, $\comp(s_i,\gamma) \cup \gamma' \models \GB{i} \rightarrow \varphi_i$ holds as well.
	\end{enumerate}
	Hence, we have shown that $\comp(s_i,\gamma) \cup \gamma' \models \GB{i} \rightarrow \varphi_i$ holds for all $\gamma \in (2^\inputs{i})^\omega$, $\gamma' \in (2^{V\setminus\variables{i}})^\omega$ and hence $s_i \models \GB{i}\rightarrow\varphi_i$ follows. Together with the previous result that $s_i \models \gb{i}$ holds, we obtain $s_i \models \gb{i} \land (\GB{i} \rightarrow \varphi_i)$, concluding the proof.\qed
\end{proof}

With \Cref{lem:soundness_certifying_synthesis,lem:completeness_certifying_synthesis}, \Cref{thm:soundness_completeness_certifying_synthesis} follows directly.


\section{Certifying Synthesis with Deterministic Certificates}\label{app:computing_guarantees}

\subsection*{Translating LTL Certificates into GTS (Proof of \Cref{lem:guar_implies_cert})}

Let $\mathcal{S} = \myVec{s_1,\dots,s_n}$, $\mathcal{G} = \myVec{g_1,\dots,g_n}$ and let $\myVec{ \varphi_1, \dots, \varphi_n}$ be the decomposition of $\varphi$. For $p_i \in \sysProc$, let $\mathcal{G}_i := \{ g_j \mid p_j \in \sysProc \setminus \{p_i\} \}$.
	Let $(\mathcal{S},\mathcal{G})$ be a solution of certifying synthesis for $\varphi$.
	Hence, both $s_i \simresp{\guarOutputs{i}} g_i$ and $s_i \satresp{\mathcal{G}_i} \varphi_i$ hold for all $p_i \in \sysProc$.
	For each process $p_i \in \sysProc$, let~$\psi_i$ be the LTL formula describing the exact behavior of the guarantee transition system $g_i$, \ie, we have 
	\[\mathcal{L}(\psi_i) = \{ \comp(g_i,\gamma) \cup \gamma' \mid \gamma \in (2^\inputs{i})^\omega, \gamma' \in (2^{V\setminus\guarVariables{i}})^\omega \}.\]
	Since the state spaces of the guarantee transition systems $g_i$ are finite, such LTL formulas~$\gb{i}$ always exists. Let $\Psi := \myVec{\gb{1}, \dots, \gb{n}}$ be the vector of the LTL certificates for all system processes and let $\GB{i} := \{ \gb{j} \mid p_j \in \sysProc \setminus \{p_i\} \}$.
	We claim that $(\mathcal{S},\Psi)$ is a solution of certifying synthesis for $\varphi$ as well.
	Thus, we need to show that for all $p_i \in \sysProc$, $s_i \models \gb{i} \land (\GB{i} \rightarrow \varphi_i)$ holds. Let $p_i \in \sysProc$.
	
	First, we prove that $s_i \models \gb{i}$ holds:
	Since $s_i \simresp{\guarOutputs{i}} g_i$ holds by assumption, every sequence of variables in $\guarOutputs{i}$ produced by~$s_i$ is also produced by $g_i$. Hence, for all $\gamma \in (2^\inputs{i})^\omega$, $\comp(s_i,\gamma) \cap \guarVariables{i} = \comp(g_i,\gamma)$ holds. Let $\gamma \in (2^\inputs{i})^\omega$. Then, $\comp(s_i,\gamma) \cup \gamma' \in \mathcal{L}(\gb{i})$ follows for all $\gamma' \in (2^{V\setminus\variables{i}})^\omega$ by construction of $\gb{i}$ and since $\guarVariables{i} \subseteq \variables{i}$ holds by definition. Hence, we have $\comp(s_i,\gamma) \cup \gamma' \models \gb{i}$ for all $\gamma' \in (2^{V\setminus\variables{i}})^\omega$.
	Therefore, $s_i \models \gb{i}$ follows since we chose an arbitrary $\gamma \in (2^\inputs{i})^\omega$.
	
	It remains to show that $s_i \models \GB{i} \rightarrow \varphi_i$ holds, \ie, that for all $\gamma \in (2^\inputs{i})^\omega$ and all $\gamma' \in (2^{V\setminus\variables{i}})^\omega$, we have $\comp(s_i,\gamma) \cup \gamma' \models \GB{i} \rightarrow \varphi_i$. 
	Let $\gamma \in (2^\inputs{i})^\omega$ and $\gamma' \in (2^{V\setminus\variables{i}})^\omega$. We distinguish two cases:
	\begin{enumerate}
		\item Let $\pref{\comp(s_i,\gamma)}{t} \cup \pref{\gamma'}{t} \in \validHistory{t}{\mathcal{G}_i}$ hold for all points in time $t$. Then, since we have $s_i \satresp{\mathcal{G}_i} \varphi_i$ by assumption, $\comp(s_i,\gamma) \cup \gamma' \models \varphi_i$ holds. Thus, by the semantics of implication, $\comp(s_i,\gamma) \cup \gamma' \models \GB{i} \rightarrow \varphi_i$ holds as well.
		\item Let there be a point in time $t$ such that $\pref{\comp(s_i,\gamma)}{t} \cup \pref{\gamma'}{t} \not\in \validHistory{t}{\mathcal{G}_i}$ holds. For the sake of readability, let $\sigma := \pref{\comp(s_i,\gamma)}{t} \cup \pref{\gamma'}{t}$. Then, there is a guarantee transition system $g_j \in \mathcal{G}_i$ and an infinite extension $\hat{\sigma}$ of $\sigma$ such that $\sigma_k \cap \outputs{j} \neq \comp(g_j,\hat{\sigma}\cap\inputs{j})_k \cap \outputs{j}$ holds for some $k$ with $1 \leq k \leq t$ by definition of valid histories. Note that since strategies cannot look into the future and since we only consider points in time up to $t$, the particular infinite extension does not matter and is only needed since computations are defined on infinite input sequences. Thus, the above holds for \emph{all} infinite extensions of $\sigma$ as well. Therefore, by construction of $\gb{j}$, $\hat{\sigma} \not\in \mathcal{L}(\gb{j})$ holds for all infinite extensions $\hat{\sigma}$ of $\sigma$, \ie, they violate $\gb{j}$. Hence, by definition of~$\GB{i}$, by the semantics of conjunction, and since $g_j \in \mathcal{G}_i$ holds, all infinite extensions $\hat{\sigma}$ of $\sigma$ violate $\GB{i}$ as well, \ie, $\hat{\sigma} \not \in \mathcal{L}(\GB{i})$. Thus, $\hat{\sigma} \not\models \GB{i}$ and therefore, by the semantics of implication, $\hat{\sigma} \models \GB{i} \rightarrow \varphi_i$ holds for all infinite extensions $\hat{\sigma}$ of $\sigma$. Clearly, by construction of $\sigma$, $\comp(s_i,\gamma) \cup \gamma'$ is an infinite extension of $\sigma$. Thus, $\comp(s_i,\gamma) \cup \gamma' \models \GB{i}\rightarrow\varphi_i$ follows.
	\end{enumerate}
	Combining the above results for both cases, we obtain $\comp(s'_i,\gamma)\cup\gamma' \models \GB{i} \rightarrow \varphi_i$ for all $\gamma \in (2^\inputs{i})^\omega$, $\gamma' \in (2^{V\setminus\variables{i}})^\omega$ and thus $s_i \models \GB{i} \rightarrow \varphi_i$ follows.

	Hence, we have both $s_i \models \gb{i}$ and $s_i \models \GB{i} \rightarrow \varphi_i$ for all system processes $p_i \in \sysProc$.
	Thus, $(\mathcal{S},\Psi)$ is indeed a solution of certifying synthesis for $\varphi$. \qed

\subsection*{Translating GTS into LTL Certificates (Proof of \Cref{lem:cert_implies_guar})}

Let $\mathcal{S} = \myVec{s_1,\dots,s_n}$, $\Psi = \myVec{\gb{1},\dots,\gb{n}}$ and let $\myVec{\varphi_1, \dots, \varphi_n}$ be the decomposition of $\varphi$.
	For all $p_i \in \sysProc$, let $\Psi_i := \{ \gb{j} \mid p_j \in \sysProc \setminus \{p_i\}\}$.
	Let $(\mathcal{S},\Psi)$ be a solution of certifying synthesis for $\varphi$.
	Hence, $s_i \models \gb{i} \land (\GB{i} \rightarrow \varphi_i)$ holds for all $p_i \in \sysProc$.
	For each $p_i \in \sysProc$, we construct a labeled transition system representing the certificate~$g_i$ from~$s_i$ as follows: $g_i$ is a copy of $s_i$, yet, the labels of $g_i$ ignore output variables $v\in\outputs{i}$ that are not contained in $\guarOutputs{i}$, \ie, $o^g_i(t,\inp{i}) = o_i(t,\inp{i}) \cap \guarOutputs{i}$ for all states~$t$ and all inputs $\inp{i} \in 2^\inputs{i}$, where $o^g_i$ is the labeling function of $g_i$ and $o_i$ is the labeling function of $s_i$. Let $\mathcal{G} := \myVec{ g_1, \dots, g_n }$ be the vector of these guarantee transition systems and let $\mathcal{G}_i := \{ g_j \mid p_j \in \sysProc \setminus \{p_i\}\}$.
	We claim that $(\mathcal{S}, \mathcal{G})$ is a solution of certifying synthesis for $\varphi$ as well. Thus, we need to show that for all $p_i \in \sysProc$, both $s_i \simresp{\guarOutputs{i}} g_i$ and $s_i \satresp{\mathcal{G}_i} \varphi_i$ hold. Let $p_i \in \sysProc$.
	
	First, we prove that $s_i \simresp{\guarOutputs{i}} g_i$ holds: By construction of the guarantee transition systems, $g_i$ and $s_i$ only differ in their labels and, in fact, the labels agree on the variables in $\guarOutputs{i}$. Since the variables in $\guarOutputs{i}$ are the only output variables that are shared by $s_i$ and $g_i$ and, in particular, $\guarOutputs{i} \subseteq \outputs{i}$ holds, $s_i \simresp{\guarOutputs{i}} g_i$ follows.

	Hence, it remains to show that $s_i \satresp{\mathcal{G}_i} \varphi_i$ holds. That is, by definition of local satisfaction, we need to show that we have $\comp(s_i,\gamma) \cup \gamma' \models \varphi_i$ for all $\gamma \in (2^\inputs{i})^\omega$, $\gamma' \in (2^{V\setminus\variables{i}})^\omega$ with $\pref{\comp(s_i,\gamma)}{t} \cup \pref{\gamma'}{t} \in \validHistory{t}{\mathcal{G}_i}$ for all points in time $t$. 
	Let $\gamma \in (2^\inputs{i})^\omega$ and $\gamma' \in (2^{V\setminus\variables{i}})^\omega$. 
	By assumption, $s_i \models \gb{i} \land \GB{i} \rightarrow \varphi_i$ holds and thus, by the semantics of conjunction, we have $s_i \models \GB{i} \rightarrow \varphi_i$ as well. Hence, in particular, $\comp(s_i,\gamma) \cup \gamma' \models \GB{i} \rightarrow \varphi_i$ holds.
	We distinguish two cases:

	\begin{enumerate}
		\item Let $\comp(s_i,\gamma)\cup\gamma'\models\GB{i}$ hold. Then, since $\comp(s_i,\gamma) \cup \gamma' \models \GB{i} \rightarrow \varphi_i$ holds, $\comp(s_i,\gamma) \cup \gamma' \models \varphi_i$ follows immediately.
		\item Let $\comp(s_i,\gamma)\cup\gamma'\not\models\GB{i}$ hold. Then, there exists a process $p_j \in \sysProc\setminus\{p_i\}$ such that $\comp(s_i,\gamma)\cup\gamma'\not\models\gb{j}$ holds. Hence, $\comp(s_i,\gamma) \cup \gamma' \not\in \mathcal{L}(\gb{j})$. For the sake of better readability, let $\sigma := \comp(s_i,\gamma)\cup\gamma'$. Then, by construction of $\gb{j}$, we have $\sigma \neq \comp(g_j,\sigma \cap \inputs{j}) \cup (\sigma \cap (V \setminus \guarVariables{j}))$. That is, intuitively, $\sigma$ does not match a computation of $g_j$. Hence, $\sigma \cap \guarOutputs{j} \neq \comp(g_j,\sigma\cap\inputs{j}) \cap \guarOutputs{j}$ holds as we have $\guarOutputs{j} \subseteq V$. Thus, in particular, there is a point in time $k$ such that $\sigma_k \cap \guarOutputs{j} \neq \comp(g_j,\sigma\cap\inputs{j})_k \cap \guarOutputs{j}$ holds and therefore, by definition of valid histories, $\sigma \not\in \validHistory{t}{\{g_j\}}$ holds for all $t>k$. Since $p_j \in \sysProc\setminus\{p_i\}$ holds, we have $g_j \in \mathcal{G}_j$ as well and thus $\comp(s_i,\gamma)\cup\gamma' \not\in \validHistory{t}{\mathcal{G}_i}$ follows for all $t > k$.
	\end{enumerate}
	Combining the above results for both cases, we directly obtain $s_i \satresp{\mathcal{G}_i} \varphi_i$.
	
	Hence, we have both $s_i \simresp{\guarOutputs{i}} g_i$ and $s_i \satresp{\mathcal{G}_i} \varphi_i$ for all system processes $p_i \in \sysProc$. Thus, $(\mathcal{S},\mathcal{G})$ is indeed a solution of certifying synthesis for $\varphi$.\qed


\section{Computing Relevant Processes}\label{app:relevant_processes}

\subsection*{Correctness of Relevant Processes (Proof of \Cref{thm:correctness_rel_proc})}

In order to prove \Cref{thm:correctness_rel_proc}, we first prove the soundness of certifying synthesis when considering relevant processes instead of all other system processes: For every solution of certifying synthesis with relevant processes, the parallel composition of the derived strategies satisfies the specification. Intuitively, this is the case since we have $\relevantProcesses{i} \subseteq \sysProc$ for all $p_i \in \sysProc\setminus\{p_i\}$.

\begin{lemma}\label{lem:soundness_rel_proc}
	Let $\varphi$ be an LTL formula. Let $\mathcal{S} = \myVec{s_1,\dots,s_n}$ be a vector of strategies for the system processes. (1) Let $\Psi$ be a vector of LTL certificates. If $(\mathcal{S},\Psi)_{\mathcal{R}}$ is a solution of certifying synthesis for $\varphi$, then $s_1 \pc \dots \pc s_n \models \varphi$ holds. (2) Let $\mathcal{G}$ be a vector of guarantee transition systems. If $(\mathcal{S},\mathcal{G})_{\mathcal{R}}$ is a solution of certifying synthesis for $\varphi$, then $s_1 \pc \dots \pc s_n \models \varphi$ holds.
\end{lemma}
\begin{proof}
	(1) Assume that $(\mathcal{S},\Psi)_{\mathcal{R}}$ is a solution of certifying synthesis for $\varphi$.
	Let $\Psi = \myVec{\gb{1},\dots,\gb{n}}$, $\GB{i} = \{ \gb{j} \mid p_j \in \sysProc \setminus \{p_i\}\}$, and $\relGB{i} = \{ \gb{j} \mid p_j \in \relevantProcesses{i}\}$.
	Then, we have $s_i \models \gb{i} \land (\relGB{i} \rightarrow \varphi_i)$ for all $p_i \in \sysProc$. By construction of the relevant processes, $\relevantProcesses{i} \subseteq \sysProc\setminus\{p_i\}$ holds and thus $\relGB{i} \subseteq \GB{i}$ follows. Hence, since $s_i \models \gb{i} \land (\relGB{i} \rightarrow \varphi_i)$ holds, $s_i \models \gb{i} \land (\GB{i} \rightarrow \varphi_i)$ follows with the semantics of conjunction and implication. Thus, $(\mathcal{S},\Psi)$ is a solution of certifying synthesis for~$\varphi$ as well and therefore $s_1 \pc \dots \pc s_n \models \varphi$ follows with \Cref{thm:soundness_completeness_certifying_synthesis}.
	
	(2) Assume that $(\mathcal{S},\mathcal{G})_{\mathcal{R}}$ is a solution of certifying synthesis for $\varphi$ and let $\mathcal{G} = \myVec{g_1,\dots,g_n}$, $\mathcal{G}_{i} = \{ g_{j} \mid p_j \in \sysProc \setminus \{p_i\}\}$, and $\mathcal{G}^{\mathcal{R}}_{i} = \{ g_{j} \mid p_j \in \relevantProcesses{i}\}$.
		Then, we have both $s_i \satresp{\mathcal{G}^{\relevantProcessesFunc}_i} \varphi_i$ and $s_i \simresp{\guarOutputs{i}} g_i$ for all $p_i \in \sysProc$. 
		Thus, $\comp(s_i,\gamma) \cup \gamma' \models \varphi_i$ holds for all $\gamma \in (2^\inputs{i})^\omega$ and $\gamma' \in (2^{V\setminus\variables{i}})^\omega$ with $\comp(s_i,\gamma)\cup\gamma' \in \validHistory{t}{\mathcal{G}^{\relevantProcessesFunc}_i}$ for all points in time $t$. 
		Since $\relevantProcesses{i} \subseteq \sysProc\setminus\{p_i\}$ holds by construction of the relevant processes, we have $\mathcal{G}^{\relevantProcessesFunc}_i \subseteq \mathcal{G}_i$ as well. Hence, $\comp(s_i,\gamma)\cup\gamma' \in \validHistory{t}{\mathcal{G}_i}$ follows. 
		Thus, $s_i \satresp{\mathcal{G}_i} \varphi_i$ holds for all $p_i \in \sysProc$ and therefore $(\mathcal{S},\mathcal{G})$ is a solution of certifying synthesis for $\varphi$ as well. Hence, by \Cref{thm:correctness_assumption_lts}, $s_1 \pc \dots \pc s_n \models \varphi$ holds.
\end{proof}

Next, we show that a slightly weaker but still sufficient notion of completeness than used before holds for certifying synthesis with relevant processes: When considering all other processes, we showed that if $s_1 \pc \dots \pc s_n \models \varphi$ holds, then there is a vector $\Psi$ of LTL certificates and a vecot $\mathcal{G}$ of guarantee transition systems such that $(\mathcal{S},\Psi)$ is a solution of certifying synthesis with LTL certificates for $\varphi$ and such that $(\mathcal{S},\mathcal{G})$ is a solution of certifying synthesis with guarantee transition systems for $\varphi$, where $\mathcal{S} = \myVec{ s_1,\dots,s_n }$.
When considering only the certificates of relevant processes, we cannot prove this property: A strategy $s_i$ may make use of a certificate of a process $p_j$ outside of $\relevantProcesses{i}$, \ie, it may violate~$\varphi_i$ on an input sequence that deviates from $g_j$ although $\varphi_i$ is satisfiable for this input. While $s_i$ is not required to satisfy $\varphi_i$ on this input, a strategy that may only consider the certificates of relevant processes, however, is. In this case, $s_i$ does not satisfy the requirements of certifying synthesis when only considering relevant certificates.
However, we can show that if if $s_1 \pc \dots \pc s_n \models \varphi$ holds, then there are \emph{some strategies} $s'_1, \dots, s'_n$ such that we can construct certificates that, together with $\mathcal{S}'=\myVec{s'_1,\dots,s'_n}$, form a solution of certifying synthesis for $\varphi$.

The main idea is to construct strategies $s'_i$ that behave on every input sequence as $s_i$ on input sequences that can occur in the parallel composition of all strategies. Since the parallel composition satisfies $\varphi$ by assumption, the strategies~$s'_i$ do so on all input sequences that match the relevant certificates. First, we show this for certifying synthesis with full strategies and LTL certificates:

\begin{lemma}\label{lem:completeness_rel_proc_ltl}
	Let $\varphi$ be an LTL formula. Let $s_1, \dots, s_n$ be strategies for the system processes. If $s_1 \pc \dots \pc s_n \models \varphi$ holds, then there exists a vector $\mathcal{S}'$ of strategies and a vector $\Psi$ of LTL certificates such that $(\mathcal{S}',\Psi)_{\mathcal{R}}$ is a solution of certifying synthesis for $\varphi$.
\end{lemma}
\begin{proof}
	Assume that $s_1 \pc \dots \pc s_n \models \varphi$ holds. Then, by \Cref{thm:soundness_completeness_certifying_synthesis}, there exists a vector $\Psi$ of LTL certificates such that $(\mathcal{S},\Psi)$ is a solution of certifying synthesis for~$\varphi$. In particular, this holds for the LTL certificates $\Psi := \myVec{\gb{1},\dots,\gb{n}}$, where~$\gb{i}$ is the LTL formula that captures the exact behavior of $s_i$, \ie, the LTL formula with $\mathcal{L}(\psi_i) = \{ \comp(s_i,\gamma)\cup\gamma' \mid \gamma \in (2^\inputs{i})^\omega, \gamma' \in (2^{V\setminus\variables{i}})^\omega \}$, since we use this construction in the completeness proof of certifying synthesis (\Cref{thm:soundness_completeness_certifying_synthesis}). In the remainder of this proof, we assume that the LTL certificates $\gb{i}$ are constructed in this way. Let $\GB{i} = \{ \gb{j} \mid p_j \in \sysProc\setminus\{p_i\} \}$.
	
	We construct strategies $s'_1, \dots, s'_n$ as follows: For each process $p_i \in \sysProc$, define \[ \comp(s'_i, (\gamma \cup \gamma') \cap \inputs{i}) := (\comp(s_1 \pc \dots \pc s_n, \gamma) \cap \outputs{i}) \cup ((\gamma \cup \gamma')\cap\inputs{i})\] for all $\gamma \in (2^\envOutputs)^\omega$ and all $\gamma' \in (2^{V\setminus\envOutputs})^\omega$. Let $\mathcal{S}' := \myVec{s'_1,\dots,s'_n}$. 
	Moreover, for each $p_i \in \sysProc$, let $\gb{i}'$ be the LTL formula that captures the exact behavior of $s'_i$, \ie, the LTL formula with $\mathcal{L}(\gb{i}') = \{ \comp(s'_i,\gamma)\cup\gamma' \mid \gamma \in (2^\inputs{i})^\omega, \gamma' \in (2^{V\setminus\variables{i}})^\omega \}$. Let $\Psi' := \myVec{\gb{1}',\dots,\gb{n}'}$ and $\relGBPrime{i} := \{ \gb{j}' \mid p_j \in \relevantProcesses{i} \}$.

	It remains to show that $(\mathcal{S}',\Psi')_{\mathcal{R}}$ is a solution of certifying synthesis for $\varphi$.
	Hence, we prove that $s'_i \models \gb{i}' \land (\relGBPrime{i} \rightarrow \varphi_i)$ holds for all system processes $p_i \in \sysProc$. Let $p_i \in \sysProc$.
	Clearly, by construction of $\gb{i}'$, we have $s'_i \models \gb{i}'$.
	Next, let $\gamma \in (2^\inputs{i})^\omega$ and $\gamma' \in (2^{V\setminus\variables{i}})^\omega$. 
	By construction of $s'_i$, we have $\comp(s'_i,\gamma) = (\sigma \cap \outputs{i}) \cup \gamma$, where $\sigma := \comp(s_1 \pc \dots \pc s_n,(\gamma \cup \gamma')\cap\envOutputs)$ and thus we have \[\comp(s'_i,\gamma) \cup \gamma' = (\sigma\cap\outputs{i}) \cup \gamma \cup \gamma'\] 
	Since $s_1 \pc \dots \pc s_n \models \varphi$ holds by assumption, we have $\sigma \models \varphi$ and hence, by the definition of specification decomposition, $\sigma \models \varphi_i$.
	Thus, if $\gamma \cup \gamma' = \sigma \cap (V\setminus\outputs{i})$, then $\comp(s'_i,\gamma)\cup\gamma' = \sigma$ and therefore $\comp(s'_i,\gamma)\cup\gamma' \models \varphi_i$ follows immediately.
	
	Otherwise, $\gamma \cup \gamma' \neq \sigma \cap (V\setminus\outputs{i})$ holds and hence there exists a system process $p_j\in\sysProc$ such that we have $\comp(s_j,\sigma\cap\inputs{j}) \cap \outputs{j} \neq (\gamma \cup \gamma') \cap \outputs{j}$. That is, intuitively, the behavior of $p_j$ in $\gamma \cup \gamma'$ differs from its behavior defined by its strategy $s_j$. We distinguish two cases:
	
	\begin{enumerate}
		\item There exists a relevant process of $p_i$ with this property, \ie, there exists a process $p_j\in\relevantProcesses{i}$ such that $\comp(s_j,\sigma\cap\inputs{j}) \cap \outputs{j} \neq (\gamma \cup \gamma') \cap \outputs{j}$ holds. Since we have $(\gamma \cup \gamma') \cap \outputs{j} = (\comp(s'_i,\gamma) \cup \gamma') \cap\outputs{j}$ by the dis\-joint\-ness of sets of output variables, $\comp(s_j,\sigma\cap\inputs{j}) \cap \outputs{j} \neq (\comp(s'_i,\gamma) \cup \gamma') \cap\outputs{j}$ follows. Thus, by construction of $\gb{j}'$, we have $\comp(s'_i,\gamma)\cup\gamma' \not\models \gb{j}'$. Since $p_j \in \relevantProcesses{i}$ holds by assumption, we have $\gb{j}' \in \relGBPrime{i}$. Hence, by the semantics of conjunction, $\comp(s'_i,\gamma)\cup\gamma' \not\models \relGBPrime{i}$ holds as well and thus $\comp(s'_i,\gamma)\cup\gamma' \models \relGBPrime{i} \rightarrow \varphi_i$ follows with the semantics of implication.
		\item There is no relevant process of $p_i$ with this property, \ie, for all $p_j \in \sysProc$ with $\comp(s_j,\sigma\cap\inputs{j}) \cap \outputs{j} \neq (\gamma \cup \gamma') \cap \outputs{j}$, we have $p_j \not\in \relevantProcesses{i}$. Let $\mathcal{P} \subseteq \sysProc \setminus (\relevantProcesses{i}\cup\{p_i\})$ be the set of all processes $p_j$ that satisfy the property. Since $p_j \not\in \relevantProcesses{i}$ holds for all $p_j \in \mathcal{P}$ by assumption, we have $\outputs{j} \cap \propositions{\varphi_i} = \emptyset$ for all $p_j \in \mathcal{P}$ by construction of the relevant processes $\relevantProcesses{i}$. Thus, $\bigcup_{p_j\in\mathcal{P}}\outputs{j} \cap \propositions{\varphi_i} = \emptyset$ and hence the satisfaction of $\varphi_i$ is not influenced by the valuations of the variables in $\bigcup_{p_j\in\mathcal{P}}\outputs{j}$. That is, for all sequences $\sigma',\sigma'' \in (2^V)^\omega$ that agree on the valuations of variables outside of $\bigcup_{p_j\in\mathcal{P}}\outputs{j}$, we have $\sigma' \models \varphi_i$ if, and only if, $\sigma'' \models \varphi_i$. By definition of $\mathcal{P}$, $\comp(s_k,\sigma\cap\inputs{k}) \cap \outputs{k} = (\gamma \cup \gamma') \cap \outputs{k}$ holds for all $p_k \in \sysProc\setminus\mathcal{P}$. Hence, since $\comp(s_k,\sigma\cap\inputs{k}) \cap \outputs{k} = \sigma \cap \outputs{k}$ holds by the definition of parallel compositions, we have $\sigma \cap \outputs{k} = (\gamma \cup \gamma')\cap\outputs{k}$ for all $p_k \in \sysProc\setminus\mathcal{P}$. Thus, $\gamma \cup \gamma'$ and $\sigma$ agree on the valuations of the variables outside of $\bigcup_{p_j\in\mathcal{P}}\outputs{j} \cup \outputs{i}$. By disjointness of the sets of outputs variables, we have $(\comp(s'_i,\gamma)\cup\gamma')\cap\outputs{\ell} = (\gamma \cup \gamma') \cap \outputs{\ell}$ for all processes $p_\ell \in \sysProc\setminus\{p_i\}$. Therefore, $\comp(s'_i,\gamma)\cup\gamma'$ and $\sigma$ agree on the valuations of the variables outside of $\bigcup_{p_j\in\mathcal{P}}\outputs{j} \cup \outputs{i}$. Moreover, $(\comp(s'_i,\gamma)\cup\gamma')\cap\outputs{i} = \sigma \cap\outputs{i}$ holds by construction of $s'_i$. Hence, $\comp(s'_i,\gamma)\cup\gamma'$ and $\sigma$ even agree on the valuations of the variables outside of $\bigcup_{p_j\in\mathcal{P}}\outputs{j}$ and thus, as shown above, we have $\sigma \models \varphi_i$ if, and only if, $\comp(s'_i,\gamma)\cup\gamma' \models \varphi_i$ holds. By assumption, we have $\sigma \models \varphi_i$ and hence $\comp(s'_i,\gamma)\cup\gamma' \models \varphi_i$ follows. Thus, we have $\comp(s'_i,\gamma)\cup\gamma' \models \relGBPrime{i} \rightarrow \varphi_i$ by the semantics of implication.
	\end{enumerate}
	Therefore, we have both $s'_i \models \gb{i}'$ and $s'_i \models \relGBPrime{i} \rightarrow \varphi_i$ for all system processes $p_i \in \sysProc$ and thus $(\mathcal{S}',\Psi')_\mathcal{R}$ is indeed a solution of certifying synthesis for $\varphi$.\qed
\end{proof}

Next, we prove this property for certifying synthesis with GTS as well:

\begin{lemma}\label{lem:completeness_rel_proc_lts}
	Let $\varphi$ be an LTL formula.
	Let $s_1, \dots, s_n$ be strategies for the system processes. If $s_1 \pc \dots \pc s_n \models \varphi$ holds, then there exists a vector $\mathcal{S}'$ of strategies and a vector $\mathcal{G}$ of guarantee transition systems such that $(\mathcal{S}',\mathcal{G})_{\mathcal{R}}$ is a solution of certifying synthesis for $\varphi$.
\end{lemma}
\begin{proof}
	Assume that $s_1 \pc \dots \pc s_n \models \varphi$ holds. Then, by \Cref{lem:completeness_rel_proc_ltl}, there exists a vector $\mathcal{S}' = \myVec{s'_1,\dots,s'_n}$ of strategies and a vector $\Psi$ of LTL certificates such that $(\mathcal{S}',\Psi)_{\mathcal{R}}$ is a solution of certifying synthesis for $\varphi$.
	In particular, this holds for the LTL certificates $\Psi := \myVec{\gb{1},\dots,\gb{n}}$, where $\gb{i}$ is the LTL formula that captures the exact behavior of $s'_i$, \ie, $\mathcal{L}(\gb{i}) = \{ \comp(s'_i,\gamma)\cup\gamma' \mid \gamma\in(2^\inputs{i})^\omega, \gamma' \in (2^{V\setminus\variables{i}})^\omega \}$ holds, since we use this construction in the (constructive) proof of \Cref{lem:completeness_rel_proc_ltl}. In the remainder of this proof, we assume that the LTL formulas $\gb{i}$ are constructed in this way. Let $\GB{i} = \{ \gb{j} \mid p_j \in \sysProc \setminus \{p_i\} \}$ and let $\relGB{i} := \{ \gb{j} \mid p_j \in \relevantProcesses{i} \}$.
	
	For each $p_i \in \sysProc$, we construct a guarantee transition system $g_i$ as follows:~$g_i$ is a copy of $s'_i$, yet, the labels of $g_i$ ignore output variables $v\in\outputs{i}$ that are not contained in $\guarOutputs{i}$, \ie, $o^g_i(t,\inp{i}) = o_i(t,\inp{i}) \cap \guarOutputs{i}$ for all states $t$ and all inputs $\inp{i} \in 2^\inputs{i}$, where $o^g_i$ is the labeling function of $g_i$ and $o_i$ is the one of~$s'_i$. Let $\mathcal{G} := \myVec{ g_1, \dots, g_n }$, let $\mathcal{G}_i := \{ g_j \mid p_j \in \sysProc \setminus \{p_i\}\}$, and let $\mathcal{G}^{\mathcal{R}}_i := \{ g_j \mid p_j \in \relevantProcesses{i}\}$.
	We claim that $(\mathcal{S}',\mathcal{G})_{\mathcal{R}}$ is a solution of certifying synthesis for $\varphi$. Thus, we need to show that both $s'_i \simresp{\guarOutputs{i}} g_i$ and $s'_i \satresp{\mathcal{G}^{\mathcal{R}}_i} \varphi_i$ hold for all system processes $p_i \in \sysProc$. Let $p_i \in \sysProc$.

	First, we show that $s'_i \simresp{\guarOutputs{i}} g_i$ holds: By construction of the guarantee transition systems, $g_i$ and $s'_i$ only differ in their labels and, in fact, the labels agree on the variables in $\guarOutputs{i}$. Since the variables in $\guarOutputs{i}$ are the only output variables that are shared by $s'_i$ and $g_i$ and, in particular, $\guarOutputs{i} \subseteq \outputs{i}$ holds, $s'_i \simresp{\guarOutputs{i}} g_i$ follows.

	It remains to show that $s'_i \satresp{\mathcal{G}^{\mathcal{R}}_i} \varphi_i$ holds, \ie, that $\comp(s'_i,\gamma) \cup \gamma' \models \varphi_i$ holds for all $\gamma \in (2^\inputs{i})^\omega$ and all $\gamma' \in (2^{V\setminus\variables{i}})^\omega$ with $\pref{\comp(s'_i,\gamma)}{t} \cup \pref{\gamma'}{t} \in \validHistory{t}{\mathcal{G}^\mathcal{R}_i}$ for all points in time $t$. Let $\gamma \in (2^\inputs{i})^\omega$ and $\gamma' \in (2^{V\setminus\variables{i}})^\omega$.	
	Since $(\mathcal{S}',\Psi)_{\mathcal{R}}$ is a solution of certifying synthesis for $\varphi$, we have $s'_i \models \gb{i} \land \relGB{i} \rightarrow \varphi_i$ and hence, by the semantics of conjunction, $s'_i \models \relGB{i} \rightarrow \varphi_i$ holds as well. Thus, in particular, $\comp(s'_i,\gamma)\cup\gamma' \models \relGB{i} \rightarrow \varphi_i$. We distinguish two cases:
	
	\begin{enumerate}
		\item Let $\comp(s'_i,\gamma)\cup\gamma' \models \relGB{i}$ hold. Then, since $\comp(s'_i,\gamma)\cup\gamma' \models \relGB{i} \rightarrow \varphi_i$ holds, $\comp(s'_i,\gamma)\cup\gamma' \models \varphi_i$ follows immediately.
		\item Let $\comp(s'_i,\gamma)\cup\gamma' \not\models \relGB{i}$ hold. Then, there exists an LTL formula $\gb{j} \in \relGB{i}$ such that $\comp(s'_i,\gamma)\cup\gamma'\not\models\gb{j}$ holds. Hence, $\comp(s'_i,\gamma) \cup \gamma' \not\in \mathcal{L}(\gb{j})$ follows. For the sake of better readability, let $\sigma := \comp(s'_i,\gamma)\cup\gamma'$. Then, by construction of $\gb{j}$, we have $\sigma \cap \guarVariables{j} \neq \comp(s'_j,\sigma \cap \inputs{j}) \cap \guarVariables{j}$. By construction of the guarantee transition systems, $g_j$ produces the same outputs as~$s'_j$ for all outputs in $\guarOutputs{j}$, \ie, $\comp(s'_j,\sigma \cap \inputs{j}) \cap \guarOutputs{j} = \comp(g_j,\sigma\cap\inputs{j}) \cap \guarOutputs{j}$. Hence, since $\guarOutputs{j} \subseteq \guarVariables{j}$, we have $\sigma \cap \guarOutputs{j} \neq \comp(g_j,\sigma \cap \inputs{j}) \cap \guarOutputs{j}$. Thus, in particular, there is a point in time $k$ such that $\sigma_k \cap \guarOutputs{j} \neq \comp(g_j,\sigma\cap\inputs{j}_k) \cap \guarOutputs{j}$ holds and therefore, by definition of valid histories, $\sigma \not\in \validHistory{t}{\{g_j\}}$ holds for all $t>k$. Since $\gb{j} \in \relGB{j}$ holds, we have $p_j \in \relevantProcesses{i}$ by definition. Thus, $g_j \in \mathcal{G}^\mathcal{R}_i$ holds as well and therefore $\comp(s_i,\gamma)\cup\gamma' \not\in \validHistory{t}{\mathcal{G}_i}$ follows for all $t > k$.
	\end{enumerate}
	Combining the above results for both cases, we directly obtain $s_i \satresp{\mathcal{G}^\mathcal{R}_i} \varphi_i$.
	
	Hence, we have both $s'_i \simresp{\guarOutputs{i}} g_i$ and $s'_i \satresp{\mathcal{G}^{\mathcal{R}}_i} \varphi_i$ for all system processes $p_i \in \sysProc$. Thus, $(\mathcal{S}',\mathcal{G})_\mathcal{R}$ is indeed a solution of certifying synthesis for $\varphi$.\qed
\end{proof}

With \Cref{lem:soundness_rel_proc,lem:completeness_rel_proc_ltl,lem:completeness_rel_proc_lts}, \Cref{thm:correctness_rel_proc} follows directly.


\section{Synthesizing Certificates}\label{app:synthesizing_certificates}

In this section, we take a closer look at the (approximative) description of certifying synthesis using local strategies.
First, we investigate the relation between the \emph{satisfaction} of an LTL formula with a \emph{local strategy} and \emph{local satisfaction} of an LTL formula with a \emph{complete strategy}. Reusing some of the results for local satisfaction from \Cref{sec:assumption_lts}, we obtain a theorem on soundness and conditional completeness of approximative certifying synthesis with local strategies.
Afterwards, we present the SAT constraint system from \Cref{thm:constraint_system} that encodes approximate certifying synthesis with local strategies.


\subsection*{Certifying Synthesis with Local Strategies}

In the following, we will first introduce the notions of \emph{restricting complete strategies} and \emph{extending local strategies}. This allows us to compare local strategies to complete ones and, in fact, to switch between local satisfaction with complete strategies and satisfaction with local strategies.

We can build a local strategy from a complete strategy by \emph{restricting} it to a set of guarantee transition systems. Intuitively, the local strategy is a copy of the full one. Yet, we delete all transitions that are only taken if the other (observable) processes deviate from their certificates. Then, the resulting strategy meets the requirements of a local strategy while still behaving exactly as the complete strategy on inputs where the other processes stick to their guaranteed behavior.

\begin{definition}[Strategy Restriction]
	Let $\mathcal{G} = \myVec{g_1,\dots,g_n}$ be a vector of GTS for the system processes. For $p_i \in \sysProc$, define $\mathcal{G}_i := \{ g_j \mid p_j \in \sysProc \setminus \{p_i\}\}$.
	Let $s_i$ be a strategy for $p_i \in \sysProc$. The restriction $\restrict{s_i}{\mathcal{G}}$ of $s_i$ to a local strategy~$s'_i$ with respect to $\mathcal{G}_i$ is defined as follows:~$s'_i$ is a copy of $s_i$, yet, for $\gamma \in (2^\inputs{i})^\omega$, $\comp(s'_i,\gamma)$ is infinite, if, and only if, there exists a sequence $\gamma' \in (2^{V\setminus\variables{i}})^\omega$ such that $\pref{\comp(s'_i,\gamma)}{t} \cup \pref{\gamma'}{t} \in \validHistory{t}{\mathcal{G}_i}$ holds for all points in time $t$.
\end{definition}

Vice versa, we can \emph{extend} a local strategy with respect to a set of guarantee transition systems $\mathcal{G}$ with its own guaranteed behavior, that is always complete, to obtain a complete strategy. The complete strategy behaves exactly as the local one on input sequences that do not deviate from the (observable) certificates of the other processes. On other input sequences, the complete strategy acts as the local one until the latter ``gets stuck'', \ie, until the local strategy does not have any outgoing transitions for the input anymore, and then switches to behaving as the local strategy's own guaranteed behavior.

\begin{definition}[Strategy Extension] 
	Let $\mathcal{G} = \myVec{g_1,\dots,g_n}$ be a vector of GTS for the system processes. For $p_i \in \sysProc$, define $\mathcal{G}_i := \{ g_j \mid p_j \in \sysProc \setminus \{p_i\}\}$.
	Let~$s_i$ be a local strategy for $p_i \in \sysProc$ with respect to $\mathcal{G}_i$. The extension $\extend{s_i}{\mathcal{G}}$ of $s_i$ to a complete strategy~$s'_i$ is defined as follows: 
	For all $\gamma \in (2^\inputs{i})^\omega$ for which $\comp(s_i,\gamma)$ is infinite, define $\comp(s'_i,\gamma) := \comp(s_i,\gamma)$. 
	For all $\gamma \in (2^\inputs{i})^\omega$ for which $\comp(s_i,\gamma)$ is finite, define $\comp(s'_i,\gamma) := \comp(g_i,\gamma) \cup \sigma'$ for some $\sigma' \in (2^{\outputs{i} \setminus \guarOutputs{i}})^\omega$.
\end{definition}

Note that whether or not a strategy extension really extends a local strategy~$s'_i$, in the sense that it behaves as~$s'_i$ until~$s'_i$ gets stuck and then switches to behave as $g_i$, strongly depends on the certificate: It has to reflect the behavior of $s'_i$ up to the point in time where the latter gets stuck. Our use of certificates and the requirements we pose on them in certifying synthesis, however, always guarantee that strategy extensions match our intuition as described above.

Using strategy extension and restriction, we can now investigate the relation between satisfaction with local strategies and local strategies with complete strategies. Every solution of certifying synthesis with local strategies can be extended to a solution of certifying synthesis with complete strategies and local satisfaction by extending the local strategies with their own guarantees behavior, \ie, by using strategy extension as defined above:

\begin{lemma}\label{lem:local_to_complete}
	Let $\varphi$ be an LTL formula. 
	Let $\mathcal{G} = \myVec{g_1, \dots, g_n}$ be a vector of guarantee transition systems for the system processes and let $\mathcal{G}_i = \{ g_j \mid p_j \in \sysProc \setminus \{p_i\}\}$. Let $\mathcal{S} = \myVec{s_1, \dots, s_n}$ be a vector of local strategies such that $s_i$ is a local strategy with respect to $\mathcal{G}_i$.	
	If $(\mathcal{S},\mathcal{G})$ is a solution of certifying synthesis with local strategies for $\varphi$, then $(\mathcal{S}',\mathcal{G})$ is a solution of certifying synthesis with local satisfaction for $\varphi$, where $\mathcal{S}' = \myVec{s'_1, \dots, s'_n}$ with $s'_i = \extend{s_i}{\mathcal{G}}$.
\end{lemma}
\begin{proof}
	Let $\myVec{\varphi_1,\dots,\varphi_n}$ be the decomposition of $\varphi$. Assume that $(\mathcal{S},\mathcal{G})$ is a solution of certifying synthesis with local strategies for $\varphi$.
	Then, for all $p_i \in \sysProc$, we have $s_i \simresp{\guarOutputs{i}} g_i$ and, for all $\gamma \in (2^\inputs{i})^\omega$, $\gamma' \in (2^{V\setminus\variables{i}})^\omega$, the run of $\mathcal{A}_i$ induced by $\comp(s_i,\gamma)\cup\gamma'$ contains only finitely many visits to rejecting states, where $\mathcal{A}_i$ is a universal co-Büchi automaton with $\mathcal{L}(\mathcal{A}_i) = \mathcal{L}(\varphi_i)$.
	To prove that $(\mathcal{S}',\mathcal{G})$ is a solution of certifying synthesis with local satisfaction for $\varphi$, we need to show that, for all $p_i \in \sysProc$, $s'_i \simresp{\guarOutputs{i}} g_i$ and $s'_i \satresp{\mathcal{G}_i} \varphi_i$ holds. Let $p_i \in \sysProc$.
	
	First, we show that $s'_i \simresp{\guarOutputs{i}} g_i$ holds: Since $s_i \simresp{\guarOutputs{i}} g_i$ holds by assumption, every sequence of variables in $\guarOutputs{i}$ produced by $s_i$ is also produced by $g_i$. Hence, for all $\gamma \in (2^\inputs{i})^\omega$ and for all points in time $t$ with $1 \leq t \leq |\comp(s_i,\gamma)|$, we have $\pref{\comp(s_i,\gamma)}{t} \cap \guarOutputs{i} = \pref{\comp(g_i,\gamma)}{t} \cap \guarOutputs{i}$ and hence, since both $\guarVariables{i} = \inputs{i} \cup \guarOutputs{i}$ and $\comp(s_i,\gamma) \cap \inputs{i} = \gamma = \comp(g_i,\gamma) \cap \inputs{i}$ hold by definition, we have \[\pref{\comp(s_i,\gamma)}{t} \cap \guarVariables{i} = \pref{\comp(g_i,\gamma)}{t} \cap \guarVariables{i}.\]
	Hence, the strategy extension $s'_i$ of $s_i$ is indeed an extension in the sense that it behaves as $s_i$ until $s_i$ gets stuck and then switches to behave like $g_i$. Let $\gamma \in (2^\inputs{i})^\omega$. We distinguish two cases:
	
	\begin{enumerate}
		\item Let $\comp(s_i,\gamma)$ be infinite. Then, $\comp(s_i,\gamma) \cap \guarVariables{i} = \comp(g_i,\gamma) \cap \guarVariables{i}$ holds since then $\pref{\comp(s_i,\gamma)}{t} \cap \guarVariables{i} = \pref{\comp(g_i,\gamma)}{t} \cap \guarVariables{i}$ holds for all points in time $t$ as shown above. By the definition of strategy extension, we have $\comp(s'_i,\gamma) = \comp(s_i,\gamma)$ since $\comp(s_i,\gamma)$ is infinite by assumption and hence $\comp(s'_i,\gamma) \cap \guarVariables{i} = \comp(g_i,\gamma) \cap \guarVariables{i}$ follows.
		\item Let $\comp(s_i,\gamma)$ be finite. Then, by definition of strategy extension and thus by construction of $s'_i$, we have $\comp(s'_i,\gamma) = \comp(g_i,\gamma) \cup \sigma'$ for some sequence $\sigma' \in (2^{\outputs{i} \setminus \guarOutputs{i}})^\omega$. Thus, in fact, $\comp(s'_i,\gamma) \cap \guarVariables{i} = \comp(g_i,\gamma) \cap \guarVariables{i}$.
	\end{enumerate}
	Thus, combining the results for both the finite and the infinite case, we have $\comp(s'_i,\gamma) \cap \guarVariables{i} = \comp(g_i,\gamma) \cap \guarVariables{i}$ for all $\gamma \in (2^\inputs{i})^\omega$. Hence, by definition of simulation, $s'_i \simresp{\guarOutputs{i}} g_i$ holds.
	
	It remains to show that $s'_i \satresp{\mathcal{G}_i} \varphi_i$ holds, \ie, that $\comp(s_i,\gamma)\cup\gamma' \models \varphi_i$ holds for all $\gamma \in (2^\inputs{i})^\omega$, $\gamma' \in (2^{V\setminus\variables{i}})^\omega$ with $\pref{\comp(s_i,\gamma)}{t} \cup \pref{\gamma'}{t} \in \validHistory{t}{\mathcal{G}_i}$ for all~$t$.
	By assumption, for all $\gamma \in (2^\inputs{i})^\omega$ and $\gamma' \in (2^{V\setminus\variables{i}})^\omega$, the run of $\mathcal{A}_i$ induced by $\comp(s_i,\gamma)\cup\gamma'$ contains only finitely many visits to rejecting states.
	Let $\gamma \in (2^\inputs{i})^\omega$, $\gamma' \in (2^{V\setminus\variables{i}})^\omega$. We distinguish two cases:
	
	\begin{enumerate}
		\item Let $\comp(s_i,\gamma)$ be infinite. Then the run of $\mathcal{A}_i$ induced by $\comp(s_i,\gamma)\cup\gamma'$ is infinite. Hence, by the definition of the co-Büchi acceptance condition, $\comp(s_i,\gamma)\cup\gamma' \models \varphi_i$. Since $\comp(s_i,\gamma)$ is infinite by assumption, we have $\comp(s_i,\gamma) = \comp(s'_i,\gamma)$ by definition of strategy extension and therefore $\comp(s'_i,\gamma)\cup\gamma' \models \varphi_i$ follows.
		\item Let $\comp(s_i,\gamma)$ be finite. Then, by the definition of local strategies, for all $\gamma'' \in (2^{V\setminus\variables{i}})^\omega$, there is a point in time $t$ such that $\pref{\comp(s_i,\gamma)}{t} \cup \pref{\gamma''}{t} \not\in \validHistory{t}{\mathcal{G}_i}$ holds. Hence, in particular, this holds for $\gamma'$, \ie, there is a point in time $t$ such that we have $\pref{\comp(s_i,\gamma)}{t} \cup \pref{\gamma'}{t} \not\in \validHistory{t}{\mathcal{G}_i}$. For the sake of readability, let $\sigma := \pref{\comp(s_i,\gamma)}{t} \cup \pref{\gamma'}{t}$. Then, there is a process $p_j$ with $g_j \in \mathcal{G}_i$ such that $\sigma_k \cap \guarOutputs{j} \neq \comp(g_j,\hat{\sigma} \cap \inputs{j})_k \cap \guarOutputs{j}$ holds for some infinite extension $\hat{\sigma}$ of $\sigma$ and some $k$ with $1 \leq k \leq t$. Since $\comp(s_i,\gamma)$ is finite by assumption, we have $\comp(s'_i,\gamma) = \comp(g_i,\gamma) \cup \sigma'$ for some $\sigma'\in(2^{\outputs{i}\setminus\guarOutputs{i}})^\omega$ by the definition of strategy extension. Furthermore, we have $t < |\comp(s_i,\gamma)|$ and hence, as shown above, $\pref{\comp(s_i,\gamma)}{t} \cap \guarVariables{i} = \pref{\comp(g_i,\gamma)}{t} \cap \guarVariables{i}$ holds. Thus, in particular $\pref{\comp(s'_i,\gamma)}{t} \cap \guarVariables{i} = \pref{\comp(s_i,\gamma)}{t} \cap \guarVariables{i}$ holds and hence $\sigma$ and $\pref{\comp(s'_i,\gamma)}{t} \cup \gamma'$ can only differ on output variables of $p_i$ outside of $\guarOutputs{i}$, \ie, on variables in $\outputs{i}\setminus\guarOutputs{i}$. Hence, by disjointness of the sets of output variables, we have $\sigma_k \cap \guarOutputs{j} = \sigma'_k \cap \guarOutputs{j}$ for all $k$ with $1\leq k \leq t$, where $\sigma' = \pref{\comp(s'_i,\gamma)}{t} \cup \pref{\gamma'}{t}$. Furthermore, the guarantee output variables are defined by $\guarOutputs{i} = \outputs{i} \cap \allInputs$. Hence, even if $\sigma$ and $\sigma'$ differ on variables in $\outputs{i}\setminus\guarOutputs{i}$, we have $\sigma \cap \inputs{j} = \sigma' \cap \inputs{j}$ and therefore, since guarantee transition systems are deterministic, \[\comp(g_j,\hat{\sigma} \cap \inputs{j})_k \cap \guarOutputs{j} = \comp(g_j,\hat{\sigma}' \cap \inputs{j})_k \cap \guarOutputs{j}\] holds for all $1 \leq k \leq t$, where $\hat{\sigma}'$ is an infinite extension of $\sigma'$. Thus, \[\sigma'_k \cap \guarOutputs{j} \neq \comp(g_j,\hat{\sigma}') \cap \inputs{j}_k \cap \guarOutputs{j}\] holds for some infinite extension $\hat{\sigma}'$ of $\sigma'$ and for some $k$ with $1 \leq k \leq t$. Therefore, $\pref{\comp(s'_i,\gamma)}{t} \cup \pref{\gamma'}{t} \not \in \validHistory{t}{\mathcal{G}_i}$ holds.
	\end{enumerate}
	Thus, combining the results for both the finite and the infinite case, we directly obtain $s'_i \satresp{\mathcal{G}_i} \varphi_i$ with the definition of local satisfaction.
	
	Therefore, we have both $s'_i \simresp{\guarOutputs{i}} g_i$ and $s'_i \satresp{\mathcal{G}_i} \varphi_i$ for all processes $p_i \in \sysProc$ and hence $(\mathcal{S}',\mathcal{G})$ is indeed a solution of certifying synthesis with complete strategies and local satisfaction for $\varphi_i$.\qed
\end{proof}

Not every solution of certifying synthesis with complete strategies and local satisfaction is one of certifying synthesis with local strategies. However, if the satisfaction of each subspecification $\varphi_i$ only depends on the variables that the corresponding process $p_i$ can observe, then we can use strategy restriction to derive a solution with local strategies from one with complete strategies:

\begin{lemma}\label{lem:complete_to_local}
	Let $\varphi$ be an LTL formula with decomposition $\myVec{\varphi_1,\dots,\varphi_n}$. Let~$\mathcal{S}$ and~$\mathcal{G}$ be vectors of strategies and GTS, respectively, for the system processes. If $(\mathcal{S},\mathcal{G})$ is a solution of certifying synthesis with local satisfaction for $\varphi$ and if $\propositions{\varphi_i}\subseteq \variables{i}$ holds for all $p_i\in\sysProc$, then $(\mathcal{S}',\mathcal{G})$ is a solution of certifying synthesis with local strategies for $\varphi$, where $\mathcal{S}' = \myVec{s'_1, \dots, s'_n}$, where $s'_i = \restrict{s_i}{\mathcal{G}}$.
\end{lemma}
\begin{proof}
	Let $\mathcal{S} = \myVec{s_1,\dots,s_n}$ and $\mathcal{G} = \myVec{g_1,\dots,g_n}$, and, for all $p_i \in \sysProc$, define $\mathcal{G}_i := \{ g_j \mid p_j \in \sysProc \setminus \{p_i\}\}$. Assume that $(\mathcal{S},\mathcal{G})$ is a solution of certifying synthesis with local satisfaction for $\varphi$.
	Then, for all $p_i \in \sysProc$, we have $s_i \simresp{\guarOutputs{i}} g_i$ and, for all $\gamma \in (2^\inputs{i})^\omega$, $\gamma' \in (2^{V\setminus\variables{i}})^\omega$ with $\pref{\comp(s_i,\gamma)}{t} \cup \pref{\gamma'}{t} \in \validHistory{t}{\mathcal{G}_i}$ for all points in time~$t$, $\comp(s_i,\gamma)\cup \gamma' \models\varphi_i$.
	Clearly, by the definition of strategy restriction, for all $p_i \in \sysProc$, $s'_i$ is a local strategy with respect to $\mathcal{G}_i$.
	To prove that $(\mathcal{S}',\mathcal{G})$ is a solution of certifying synthesis with local satisfaction for $\varphi$, we need to show that, for all $p_i \in \sysProc$, $s'_i \simresp{\guarOutputs{i}} g_i$ holds and that for all $\gamma \in (2^\inputs{i})^\omega$, $\gamma' \in (2^{V\setminus\variables{i}})^\omega$, the run of $\mathcal{A}_i$ induced by $\comp(s'_i,\gamma)\cup\gamma'$ contains only finitely many visits to rejecting states, where $\mathcal{A}_i$ is a universal co-Büchi with $\mathcal{L}(\mathcal{A}_i) = \mathcal{L}(\varphi_i)$. Let $p_i$ be a system process.

	First, we show that $s'_i \simresp{\guarOutputs{i}} g_i$ holds: Since we have $s_i \simresp{\guarOutputs{i}} g_i$, every sequence of variables in $\guarOutputs{i}$ produced by $s_i$ is also produced by $g_i$, \ie, for all $\gamma \in (2^\inputs{i})^\omega$, $\comp(s_i,\gamma)\cap\guarOutputs{i} = \comp(g_i,\gamma)\cap\guarOutputs{i}$. By the definition of strategy restriction, every sequence of variables in $\outputs{i}$ produced by $s'_i$ is also produced by $s_i$. Hence, as $s'_i$ may produce finite computations, $\pref{\comp(s'_i,\gamma)}{t} = \pref{\comp(s_i,\gamma)}{t}$ holds for all points in time $t$ with $1 \leq t \leq |\comp(s'_i,\gamma)|$. Since $\guarOutputs{i} \subseteq \outputs{i} \subseteq \variables{i}$ holds by definition, we thus have $\pref{\comp(s'_i,\gamma)}{t} \cap \guarOutputs{i} = \pref{\comp(g_i,\gamma)}{t} \cap \guarOutputs{i}$ for all $\gamma \in (2^\inputs{i})^\omega$ and for all $t$ with $1 \leq t \leq |\comp(s'_i,\gamma)|$. Thus, $s'_i \simresp{\guarOutputs{i}} g_i$ follows.
	
	Second, we show that for all $\gamma \in (2^\inputs{i})^\omega$, $\gamma' \in (2^{V\setminus\variables{i}})^\omega$, the run of $\mathcal{A}_i$ induced by $\comp(s'_i,\gamma)\cup\gamma'$ contains only finitely many visits to rejecting states: Let $\gamma \in (2^\inputs{i})^\omega$, $\gamma' \in (2^{V\setminus\variables{i}})^\omega$. We distinguish two cases:
	
	\begin{enumerate}
		\item Let $\pref{\comp(s_i,\gamma)}{t}\cup\pref{\gamma'}{t} \in \validHistory{t}{\mathcal{G}_i}$ hold for all points in time $t$. Then, since $s_i \satresp{\mathcal{G}_i} \varphi_i$ holds, we have $\comp(s_i,\gamma)\cup\gamma'\models\varphi_i$. Thus, by definition of the co-Büchi acceptance condition and since $\mathcal{L}(\mathcal{A}_i) = \mathcal{L}(\varphi_i)$ holds, the run of~$\mathcal{A}_i$ induced by $\comp(s_i,\gamma)\cup\gamma'$ contains only finitely many visits to rejecting states. Since $\pref{\comp(s_i,\gamma)}{t}\cup\pref{\gamma'}{t} \in \validHistory{t}{\mathcal{G}_i}$ holds for all points in time $t$ by assumption, $s'_i$ is a copy of $s_i$ by the definition of strategy restriction. Thus, $\comp(s'_i,\gamma)\cup\gamma' = \comp(s_i,\gamma)\cup\gamma'$ holds and hence the run of~$\mathcal{A}_i$ induced by $\comp(s'_i,\gamma)\cup\gamma'$ contains only finitely many visits to rejecting states.
		\item Let there be a point in time $t$ such that $\pref{\comp(s_i,\gamma)}{t}\cup\pref{\gamma'}{t} \not\in \validHistory{t}{\mathcal{G}_i}$ holds. If for all $\gamma'' \in (2^{V\setminus\variables{i}})^\omega$, there is a $t$ such that $\pref{\comp(s_i,\gamma)}{t}\cup\pref{\gamma''}{t} \not\in \validHistory{t}{\mathcal{G}_i}$ holds, then, by the definition of strategy restriction, $\comp(s'_i,\gamma)$ is finite. Hence, the run of $\mathcal{A}_i$ induced by $\comp(s'_i,\gamma)\cup\gamma'$ is finite as well and therefore it contains only finitely many visits to rejecting states. Otherwise, there is a $\gamma'' \in (2^{V\setminus\variables{i}})^\omega$ such that $\pref{\comp(s_i,\gamma)}{t}\cup\pref{\gamma''}{t} \in \validHistory{t}{\mathcal{G}_i}$ holds for all points in time~$t$. Then, since $s_i \satresp{\mathcal{G}_i} \varphi_i$ holds, we have $\comp(s_i,\gamma)\cup\gamma''\models\varphi_i$. By assumption, we have $\propositions{\varphi_i} \subseteq \variables{i}$ and hence the satisfaction of $\varphi_i$ is independent of the valuations of the variables in $V\setminus\variables{i}$. Hence, $\comp(s_i,\gamma)\cup\gamma' \models\varphi_i$ follows as well. Therefore, by the co-Büchi condition and since  $\mathcal{L}(\mathcal{A}_i) = \mathcal{L}(\varphi_i)$ holds, the run of~$\mathcal{A}_i$ induced by $\comp(s_i,\gamma)\cup\gamma'$ contains only finitely many visits to rejecting states. Since $\pref{\comp(s_i,\gamma)}{t}\cup\pref{\gamma''}{t} \in \validHistory{t}{\mathcal{G}_i}$ holds for all points in time~$t$, $s'_i$ is a copy of $s_i$ by the definition of strategy restriction. Thus, $\comp(s'_i,\gamma)\cup\gamma' = \comp(s_i,\gamma)\cup\gamma'$ holds and hence the run of~$\mathcal{A}_i$ induced by $\comp(s'_i,\gamma)\cup\gamma'$ contains only finitely many visits to rejecting states.
	\end{enumerate}
	
	Hence, for all $p_i\in\sysProc$, we have $s'_i \simresp{\guarOutputs{i}} g_i$ and, for all $\gamma \in (2^\inputs{i})^\omega$, $\gamma' \in (2^{V\setminus\variables{i}})^\omega$, the run of $\mathcal{A}_i$ induced by $\comp(s'_i,\gamma)\cup\gamma'$ contains only finitely many visits to rejecting states. Thus, $(\mathcal{S}',\mathcal{G})$ is indeed a solution of certifying synthesis with local strategies for $\varphi$.\qed
\end{proof}

With these results as well as the soundness and completeness of certifying synthesis with complete strategies, guarantee transition systems, and local satisfaction, we thus obtain the following theorem:

\begin{theorem}\label{thm:corectness_local_strategies}
	Let $\varphi$ be an LTL formula.
	\begin{enumerate}
		\item Let $\mathcal{G} = \myVec{g_1,\dots,g_n}$ be a vector of GTS and let $\mathcal{G}_i = \{ g_j \mid p_j \in \sysProc \}$. Let $\mathcal{S} = \myVec{s_1,\dots,s_n}$ be a vector of local strategies such that for all $p_i \in \sysProc$, $s_i$ is a local strategy with respect to~$\mathcal{G}_i$. If $(\mathcal{S},\mathcal{G})$ is a solution of certifying synthesis with local strategies for $\varphi$, then $s_1 \pc \dots \pc s_n \models \varphi$ holds.
		\item Let $\mathcal{S} = \myVec{s_1,\dots,s_n}$ be a vector of strategies and let $\myVec{\varphi_1,\dots,\varphi_n}$ be the decomposition of $\varphi$. If $s_1 \pc \dots \pc s_n \models \varphi$ and if $\propositions{\varphi_i} \subseteq \variables{i}$ holds for all $p_i\in\sysProc$, then there exists a vector $\mathcal{G}$ of guarantee transition systems, such that $(\mathcal{S}',\mathcal{G})$ is a solution of certifying synthesis with local strategies for $\varphi$, where $\mathcal{S'} = \myVec{s'_1,\dots,s'_n}$ with $s'_i  = \restrict{s_i}{\mathcal{G}}$.
	\end{enumerate}
\end{theorem}
\begin{proof}
	First, let $\mathcal{G} = \myVec{g_1,\dots,g_n}$ be a vector of GTS, let $\mathcal{S} = \myVec{s_1,\dots,s_n}$ be a vector of local strategies such that for all $p_i \in \sysProc$, $s_i$ is a local strategy with respect to~$\mathcal{G}_i$, and assume that $(\mathcal{S},\mathcal{G})$ is a solution of certifying synthesis with local strategies for $\varphi$. Then, by \Cref{lem:local_to_complete}, $(\mathcal{S}',\mathcal{G})$ is a solution of certifying synthesis with local satisfaction and guarantee transition systems for~$\varphi$, where $\mathcal{S}' = \myVec{s'_1,\dots,s'_n}$ with $s'_i = \extend{s_i}{\mathcal{G}}$. Thus, $s'_1 \pc \dots \pc s'_n \models \varphi$ follows with \Cref{thm:correctness_assumption_lts}. 
	It remains to show that $\comp(s'_1 \pc \dots \pc s'_n,\gamma) = \comp(s_1 \pc \dots \pc s_n,\gamma)$ holds for all $\gamma \in (2^\envOutputs)^\omega$:
	
	Towards a contradiction, suppose that there is a $\gamma \in (2^\envOutputs)^\omega$ such that $\comp(s'_1 \pc \dots \pc s'_n,\gamma) \neq \comp(s_1 \pc \dots \pc s_n,\gamma)$ holds. For the sake of readability, let $\sigma = \comp(s_1 \pc \dots \pc s_n,\gamma)$ and $\sigma' = \comp(s'_1 \pc \dots \pc s'_n,\gamma)$. Let $t$ be the earliest point in time at which $\sigma$ and $\sigma'$ differ. Then, there is a process $p_i \in \sysProc$ such that $\pref{\comp(s'_i,\sigma'\cap\inputs{i})}{t} \cap \outputs{i} \neq \pref{\comp(s_i,\sigma\cap\inputs{i})}{t} \cap \outputs{i}$ holds while we have $\sigma_k = \sigma'_k$ for all $k$ with $1 \leq k < t$. Since strategies are represented by Moore transition systems, they cannot react to an input directly. Hence, since $t$ is the earliest point in time at which $\sigma$ and $\sigma'$ differ by assumption, we have \[\pref{\comp(s'_i,\sigma\cap\inputs{i})}{t} \cap \outputs{i} \neq \pref{\comp(s_i,\sigma\cap\inputs{i})}{t} \cap \outputs{i}\] as well since $\pref{\comp(s'_i,\sigma'\cap\inputs{i})}{t} \cap \outputs{i} = \pref{\comp(s'_i,\sigma\cap\inputs{i})}{t} \cap \outputs{i}$ holds. By definition of strategy extension, this is only possible if the transition system representing $s_i$ gets stuck at point in time $t$ on input $\sigma \cap \inputs{i}$. Hence, $\comp(s_i,\sigma \cap \inputs{i})$ is finite.
	Therefore, by definition of local strategies, for all $\gamma' \in (2^{V\setminus\variables{i}})^\omega$, there is a point in time $k$ such that $\pref{\comp(s_i,\sigma\cap\inputs{i})}{k} \cup \pref{\gamma'}{k} \not\in\validHistory{k}{\mathcal{G}_i}$, where $\mathcal{G}_i = \mathcal{G}\setminus\{g_i\}$. Thus, in particular, $\pref{\comp(s_i,\sigma\cap\inputs{i})}{k} \cup (\pref{\sigma}{k} \cap (V\setminus\variables{i})) \not\in\validHistory{k}{\mathcal{G}_i}$ holds. Note that we have $\pref{\comp(s_i,\sigma\cap\inputs{i})}{k} \cup (\pref{\sigma}{k} \cap (V\setminus\variables{i})) = \sigma$ and hence $\sigma \not\in\validHistory{k}{\mathcal{G}_i}$ holds. Thus, by the definition of valid histories, there is a system process $p_j \in \sysProc\setminus\{p_i\}$ and a point in time $\ell$ with $1 \leq \ell \leq k$ such that $\sigma_\ell \cap \guarOutputs{j} \neq \comp(g_j,\hat{\sigma}\cap\inputs{j})_\ell \cap \guarOutputs{j}$ holds for an infinite extension $\hat{\sigma}$ of $\pref{\sigma}{k}$. Note that this indeed holds for all infinite extensions of $\pref{\sigma}{k}$ since strategies cannot look into the future and $\ell \leq k$. Hence, this holds in particular for $\hat{\sigma} := \sigma$. By construction of $\sigma$ and the definition of parallel composition, however, $\pref{\sigma}{\ell} \cap \outputs{j} = \pref{\comp(s_j,\sigma\cap\inputs{j})}{\ell} \cap \outputs{j}$ holds. Hence, since $\guarOutputs{j} \subseteq \outputs{j}$, we obtain $\pref{\comp(s_j,\sigma\cap\inputs{j})}{\ell} \cap \guarOutputs{j} \neq \comp(g_j,\sigma\cap\inputs{j})_\ell \cap \guarOutputs{j}$.
	Since $(\mathcal{S},\mathcal{G})$ is a solution of certifying synthesis by assumption, we have $s_i \simresp{\guarOutputs{i}} g_i$. Hence, every sequence of valuations of variables in $\guarOutputs{i}$ produced by $s_i$ is also produced by $g_i$. Therefore, in particular, $\pref{\comp(s_j,\sigma\cap\inputs{j})}{\ell} \cap \guarVariables{j} = \comp(g_j,\sigma\cap\inputs{j})_\ell \cap \guarVariables{j}$ holds and thus $\pref{\comp(s_j,\sigma\cap\inputs{j})}{\ell} \cap \guarOutputs{j} = \comp(g_j,\sigma\cap\inputs{j})_\ell \cap \guarOutputs{j}$ follows since $\guarOutputs{j} \subseteq \guarVariables{j}$ holds, yielding a contradiction. Thus, we have \[\comp(s'_1 \pc \dots \pc s'_n,\gamma) = \comp(s_1 \pc \dots \pc s_n,\gamma)\] for all $\gamma \in (2^\envOutputs)^\omega$ and hence, since $s'_1 \pc \dots \pc s'_n \models \varphi$ holds as shown above, $s_1 \pc \dots \pc s_n \models \varphi$ follows.
	
	Second, let $\mathcal{S} = \myVec{s_1,\dots,s_n}$ be a cevtor of strategies and assume that both $s_1 \pc \dots \pc s_n \models \varphi$ and $\propositions{\varphi_i} \subseteq \variables{i}$ hold for all $p_i\in\sysProc$. Then, by \Cref{thm:correctness_assumption_lts}, there exists a vector $\mathcal{G}$ of guarantee transition systems such that $(\mathcal{S},\mathcal{G})$ is a solution of certifying synthesis for $\varphi$. Hence, since $\propositions{\varphi_i} \subseteq \variables{i}$ holds for all $p_i\in\sysProc$ by assumption, it directly follows with \Cref{lem:complete_to_local} that $(\mathcal{S}',\mathcal{G})$ is a solution of certifying synthesis, proving the claim.\qed
\end{proof}


\subsection*{Constraint System}

\input{constraint_system}


\section{Experimental Evaluation}\label{app:benchmarks}

\begin{table}[t]
    \centering
    \caption{Experimental results on scalable benchmarks. Reported is the parameter and the running time in seconds. We used a machine with a 3.1 GHz Dual-Core Intel Core i5 processor and 16 GB of RAM, and a timeout of 60 min. For dist.\ BoSy, we use the SMT encoding and give the average runtime of 10~runs.\\}\label{table:results_extended}
    \begin{tabular}{p{3.2cm}>{\centering}p{1.3cm}||>{\centering}p{2.1cm}|>{\centering}p{2.1cm}|>{\centering\arraybackslash}p{2.1cm}}
	     Benchmark & Param. & Cert. Synth. & Dist. BoSy & Dom. Strat.\\
	     \hline\hline
	     n-ary Latch & 2 & \textbf{0.89} & 41.26 & 4.75\\
	     & 3 & \textbf{0.91} & TO & 6.40\\
	     & 4 & \textbf{0.92} & TO & 8.46\\
	     & 5 & \textbf{0.94} & TO & 10.74\\
	     & 6 & \textbf{12.26} & TO & 13.89\\
	     & 7 & 105.69 & TO & \textbf{15.06}\\
	     \hline
	     Shift & 2 & \textbf{1.10} & 1.99 & 4.76 \\
	     & 3 & \textbf{1.13} & 4.16 & 7.04 \\
	     & 4 & \textbf{1.14} & TO & 11.13 \\
	     & 5 & \textbf{1.29} & TO & 13.68 \\
	     & 6 & \textbf{2.20} & TO & 16.01 \\
	     & 7 & \textbf{9.01} & TO & 16.08 \\
	     & 8 & 71.89 & TO & \textbf{19.38} \\
	     \hline
	     Manufacturing Robots & 2 & \textbf{1.10} & 2.45 & --\\
	     & 4 & \textbf{1.18} & 2.43 & --\\
	     & 6 & \textbf{1.67} & 3.20 & --\\
	     & 8 & \textbf{2.88} & 5.67 & --\\
	     & 10 & \textbf{48.83} & 221.16 & --\\
	     & 12 & \textbf{1.44} & TO & --\\
	     & 14 & \textbf{76.32} & TO & --\\
	     & 18 & \textbf{2716.27} & TO & --\\
	     & 20 & \textbf{9.80} & TO & -- \\
	     & 24 & \textbf{8.82} & TO & -- \\
	     & 30 & \textbf{32.83} & TO & --\\
	     & 36 & \textbf{2911.26} & TO & --\\
	     & 42 & \textbf{373.90} & TO & --\\
	     & 45 & TO & TO & --
    \end{tabular}
\end{table}

In this section, we present more details on our experimental evaluation.
In \Cref{table:results} in \Cref{sec:experiments}, the results for some parameters are omitted. The full table for the affected benchmarks is presented in \Cref{table:results_extended}.

In the remainder of this section, we first give the system architectures of all benchmarks. Second, we present the specifications of the \emph{Ripple-Carry Adder} and the \emph{Manufacturing Robots} benchmark. The other benchmarks stem from the synthesis competition SYNTCOMP~\cite{SYNTCOMP2018}. Hence, we refer to the SYNTCOMP specification description for their specifications. Third, we give a more detailed description of the \emph{Manufacturing Robots} benchmark, including a table with detailed experimental results for all parameters.

\subsection*{System Architectures}

In the following, we present the system architectures of all benchmarks. The environment process $\env$ is depicted in gray, the system process are depicted in white. Incoming edges are labeled with the input variables of the process, outgoing edges with the output variables.

\paragraph{$n$-ary Latch.}

\begin{center}
\scalebox{0.92}{
\begin{tikzpicture}[>=latex,shorten >=0pt,auto,->,node distance=1cm,thin,every edge/.style={draw,font=\small}, initial text = , proc/.style={draw,rectangle,minimum width=0.9cm, minimum height=0.7cm, rounded corners = 5pt}, env/.style={draw,rectangle,minimum width=0.9cm, minimum height=0.7cm,fill=black!20}]
		
	\node[env]		(env)	at (3,1.8)		{$\mathit{env}$};
	\node[proc]		(p1)		at (0,0)		{$p_1$};
	\node[proc]		(p2)		at (2,0)		{$p_2$};
	\node			(pd)		at (4,0)		{$\dots$};
	\node[proc]		(pn)		at (6,0)		{$p_n$};
			
	\path	(env)	edge[bend right=25]	node	[left,align=center]	{$\mathit{inp}_1$ \\ $\mathit{upd}~~~$}	(p1)
					edge[bend right=5]	node[left,align=center]	{$\mathit{inp}_2~$ \\ $\mathit{upd}~~~$}	(p2)
					edge[bend left=25]	node	[right,align=center]	{$\mathit{inp}_n~$ \\ $~~~\mathit{upd}~~~$}	(pn);
					
	\node		(d1)		at (0,-1	.2)	{};
	\node		(d2)		at (2,-1	.2)	{};
	\node		(dn)		at (6,-1	.2)	{};
	
	\path	(p1)	edge 	node	[left]	{$\mathit{out}_1$}	(d1)
			(p2)	edge 	node	[left]	{$\mathit{out}_2$}	(d2)
			(pn)	edge 	node	[right]	{$\mathit{out}_n$}	(dn);
\end{tikzpicture}}
\end{center}

\paragraph{Generalized Buffer.}

\begin{center}
\scalebox{0.92}{
\begin{tikzpicture}[>=latex,shorten >=0pt,auto,->,node distance=1cm,thin,every edge/.style={draw,font=\small}, initial text = , proc/.style={draw,rectangle,minimum width=0.9cm, minimum height=0.7cm, rounded corners = 5pt}, env/.style={draw,rectangle,minimum width=0.9cm, minimum height=0.7cm,fill=black!20}]
		
	\node[env]		(env)	at (3,3)		{$\mathit{env}$};
	\node[proc]		(p1)		at (1,0)		{$p_1$};
	\node[proc]		(p2)		at (5,0)		{$p_2$};
			
	\path	(env)	edge[bend right=25]	node	[left,align=center]	{$~~~~~~~\mathit{stob\_REQ}_0$ \\ $~~~~\mathit{stob\_REQ}_1$ \\ $~\mathit{rtob\_ACK}_0$ \\ $\dots$ \\ $\mathit{rtob\_ACK}_k~~~~$}	(p1)
					edge[bend left=25]	node[right,align=center]	{$\mathit{stob\_REQ}_0~~~~~~$ \\ $\mathit{stob\_REQ}_1~~~~$ \\ $\mathit{rtob\_ACK}_0~$ \\ $\dots$ \\ $~~~~\mathit{rtob\_ACK}_k$}	(p2);
									
	\node		(d1)		at (1,-2	.2)	{};
	\node		(d2)		at (5,-2.2)	{};
	\path	(p1)	edge 	node	[left,align=center]	{$\mathit{btos\_ACK}_0$ \\ $\mathit{btos\_ACK}_1$}	(d1)
			(p2)	edge 	node	[right,align=center]	{$\mathit{btor\_REQ}_0$ \\ $\dots$ \\ $\mathit{btor\_REQ}_k$}	(d2);
\end{tikzpicture}}
\end{center}

\paragraph{Load Balancer.}

\begin{center}
\scalebox{0.92}{
\begin{tikzpicture}[>=latex,shorten >=0pt,auto,->,node distance=1cm,thin,every edge/.style={draw,font=\small}, initial text = , proc/.style={draw,rectangle,minimum width=0.9cm, minimum height=0.7cm, rounded corners = 5pt}, env/.style={draw,rectangle,minimum width=0.9cm, minimum height=0.7cm,fill=black!20}]
		
	\node[env]		(env)	at (4.5,2.5)		{$\mathit{env}$};
	\node[proc]		(p1)		at (0,0)		{$p_1$};
	\node[proc]		(p2)		at (3,0)		{$p_2$};
	\node			(pd)		at (6,0)		{$\dots$};
	\node[proc]		(pn)		at (9,0)		{$p_n$};
			
	\path	(env)	edge[bend right=25]	node	[left,align=center]	{$r_1,\dots,r_n$ \\ $job~~~$}	(p1)
					edge[bend right=5]	node[left,align=center]	{$r_1,\dots,r_n$ \\ $job~~~$}	(p2)
					edge[bend left=25]	node	[right,align=center]	{$r_1,\dots,r_n$ \\ $~~~job$}	(pn);

	\path	(p1)	edge[bend left=20] 	node		{$g_1$}	(p2)
					edge[bend right=40] 	node	[below]	{$g_1$}	(pd)
					edge[bend right=47] 	node	[below]	{$g_1$}	(pn)
			(p2)	edge[bend left=20]	node		{$g_2$}	(p1)
					edge[bend left=20]	node		{$g_2$}	(pd)
					edge[bend right=40] 	node	[below]	{$g_2$}	(pn)
			(pd)	edge[bend left=20]	node		{$g_{n-1}$}	(pn)
					edge[bend left=20]	node		{$g_3$}	(p2)
			(pn)	edge[bend left=20] 	node		{$g_n$}	(pd)
					edge[bend left=60] 	node	[below]	{$g_n$}	(p1);
\end{tikzpicture}}
\end{center}

\paragraph{Shift.}

\begin{center}
\scalebox{0.92}{
\begin{tikzpicture}[>=latex,shorten >=0pt,auto,->,node distance=1cm,thin,every edge/.style={draw,font=\small}, initial text = , proc/.style={draw,rectangle,minimum width=0.9cm, minimum height=0.7cm, rounded corners = 5pt}, env/.style={draw,rectangle,minimum width=0.9cm, minimum height=0.7cm,fill=black!20}]
		
	\node[env]		(env)	at (3.75,1.8)	{$\mathit{env}$};
	\node[proc]		(p1)		at (0,0)		{$p_1$};
	\node[proc]		(p2)		at (2.5,0)		{$p_2$};
	\node			(pd)		at (5,0)		{$\dots$};
	\node[proc]		(pn)		at (7.5,0)		{$p_n$};
			
	\path	(env)	edge[bend right=25]	node	[left]	{$i_1, \dots, i_n~~$}	(p1)
					edge[bend right=5]	node[left]	{$i_1, \dots, i_n~$}		(p2)
					edge[bend left=25]	node	[right]	{$~~i_1, \dots, i_n$}	(pn);
					
	\node		(d1)		at (0,-1	.5)	{};
	\node		(d2)		at (2.5,-1.5)	{};
	\node		(dn)		at (7.5,-1.5)	{};
	
	\path	(p1)	edge 	node	[left]	{$o_1$}	(d1)
			(p2)	edge 	node	[left]	{$o_2$}	(d2)
			(pn)	edge 	node	[right]	{$o_n$}	(dn);
\end{tikzpicture}}
\end{center}

\paragraph{Ripple-Carry Adder.}

\begin{center}
\scalebox{0.92}{
\begin{tikzpicture}[>=latex,shorten >=0pt,auto,->,node distance=1cm,thin,every edge/.style={draw,font=\small}, initial text = , proc/.style={draw,rectangle,minimum width=0.9cm, minimum height=0.7cm, rounded corners = 5pt}, env/.style={draw,rectangle,minimum width=0.9cm, minimum height=0.7cm,fill=black!20}]
		
	\node[env]		(env)	at (3.75,1.8)		{$\mathit{env}$};
	\node[proc]		(p1)		at (0,0)		{$p_1$};
	\node[proc]		(p2)		at (2.5,0)		{$p_2$};
	\node			(pd)		at (5,0)		{$\dots$};
	\node[proc]		(pn)		at (7.5,0)		{$p_n$};
			
	\path	(env)	edge[bend right=20]	node	[left,align=center]	{$x_1, y_1~~$}	(p1)
					edge[bend right=5]	node[left,align=center]	{$x_2,y_2$}	(p2)
					edge[bend left=20]	node	[right,align=center]	{$~~~x_n,y_n$}	(pn);
					
	\node		(d1)		at (0,-1	.5)		{};
	\node		(d2)		at (2.5,-1.5)	{};
	\node		(dd)		at (5,-1.5)	{};
	\node		(dn)		at (7.5,-1.5)	{};

	\path	(p1)	edge		node			{$c_1$}	(p2)
					edge		node[left]	{$s_1$}	(d1)
			(p2)	edge		node			{$c_2$}	(pd)
					edge		node[left]	{$s_2$}	(d2)
			(pd)	edge		node			{$c_{n-1}$}	(pn)
			(pn)	edge		node[left]	{$s_n, c_n$}	(dn);
\end{tikzpicture}}
\end{center}

\paragraph{Manufacturing Robots.}

\begin{center}
\scalebox{0.92}{
\begin{tikzpicture}[>=latex,shorten >=0pt,auto,->,node distance=1cm,thin,every edge/.style={draw,font=\small}, initial text = , proc/.style={draw,rectangle,minimum width=0.9cm, minimum height=0.7cm, rounded corners = 5pt}, env/.style={draw,rectangle,minimum width=0.9cm, minimum height=0.7cm,fill=black!20}]
		
	\node[env]		(env)	at (3,2)		{$\mathit{env}$};
	\node[proc]		(p1)		at (1,0)		{$p_1$};
	\node[proc]		(p2)		at (5,0)		{$p_2$};
			
	\path	(env)	edge[bend right=25]	node	[left,align=center]	{$\mathit{at\_crossing}_1$ \\ $\mathit{at\_crossing}_2~~~$}	(p1)
					edge[bend left=25]	node[right,align=center]	{$\mathit{at\_crossing}_1$ \\ $~~~\mathit{at\_crossing}_2$}	(p2)
			(p1)	edge[bend left=25]		node		{$\mathit{go}_1$}	(p2)
			(p2)	edge[bend left=25]		node		{$\mathit{go}_2$}	(p1);
									
	\node		(d1)		at (1,-1.5)	{};
	\node		(d2)		at (5,-1.5)	{};
	\path	(p1)	edge 	node	[left]	{$m_1$}	(d1)
			(p2)	edge 	node	[right]	{$m_2$}	(d2);
\end{tikzpicture}}
\end{center}

\subsection*{Benchmark Specifications}

Next, we give the specifications of the \emph{Ripple-Carry Adder} and the \emph{Manufacturing Robots} benchmarks. For the specifications of the \emph{$n$-ary Latch}, the \emph{Generalized Buffer}, the \emph{Load Balancer}, and the \emph{Shift} benchmarks, we refer to the benchmark descriptions of the synthesis competition SYNTCOMP~\cite{SYNTCOMP2018}.

\paragraph{Ripple-Carry Adder.}
The \emph{Ripple-Carry Adder} benchmark describes an adder that adds two bit vectors, both with $n$ bits. The inputs are the very first carry bit $c_\mathit{in}$ and the bits of the two bit vectors, $x_0, \dots, x_{n-1}$ and $y_0,\dots,y_{n-1}$. The outputs are the sum bits $s_0,\dots,s_{n-1}$ as well as the carry bits $c_0,\dots,c_{n-1}$ for every bit. The specification $\varphi$ is then given by $\varphi := \varphi_\mathit{init} \land \bigwedge_{0 < i < n} \varphi_i$, where

\begin{align*}
	\varphi_\mathit{init} := &\Globally \left(\Next c_0 \leftrightarrow \left( \left( x_0 \land y_0 \right) \lor \left( c_\mathit{in} \land \left( \left( x_0 \land \neg y_0 \right) \lor \left( \neg x_0 \land y_0 \right) \right) \right) \right) \right) \land \\
		&\Globally \left( \Next s_0 \leftrightarrow \left( \left( x_0 \land \neg y_0 \land \neg c_\mathit{in} \right) \lor \left( \neg x_0 \land y_0 \land \neg c_\mathit{in} \right)\right. \right. \\
		&~~~~~~~~~~~~~~~~ \left.\left.\lor \left( \neg x_0 \land \neg y_0 \land c_\mathit{in} \right) \lor \left( x_0 \land y_0 \land c_\mathit{in} \right) \right) \right)
\end{align*}
\begin{align*}
	\varphi_i := &\Globally \left(\Next c_i \leftrightarrow \left( \left( x_i \land y_i \right) \lor \left( c_{i-1} \land \left( \left( x_i \land \neg y_i \right) \lor \left( \neg x_i \land y_i \right) \right) \right) \right) \right) \land \\
		&\Globally \left( \Next s_i \leftrightarrow \left( \left( x_i \land \neg y_i \land \neg c_{i-1} \right) \lor \left( \neg x_i \land y_i \land \neg c_{i-1} \right)\right. \right. \\
		&~~~~~~~~~~~~~~~~ \left.\left.\lor \left( \neg x_i \land \neg y_i \land c_{i-1} \right) \lor \left( x_i \land y_i \land c_{i-1} \right) \right) \right)
\end{align*}

\paragraph{Manufacturing Robots.}
The \emph{Manufacturing Robots} benchmarks describes the robots from \Cref{sec:motivating_example}. It is parameterized in the additional objectives $\varphi_{\mathit{add}_i}$ of the robots. The specification has two parameters, $n_1$ and $n_2$. The additional objectives of the robots state that $r_i$ needs to visit the machine it is responsible for in every $n_i$-th step. Thus, as described in \Cref{sec:motivating_example}, the inputs are $\mathit{at\_crossing}_1$ and $\mathit{at\_crossing}_2$, describing that the corresponding robot is at the crossing. The outputs are $\mathit{go}_1$ and $\mathit{go}_2$, describing that the corresponding robots moves forward, as well as $m_1$ and $m_2$, describing that the corresponding robot reaches the machine it is responsible for. The full specification $\varphi$ is then given by $\varphi := \varphi_\mathit{safe} \land \bigwedge_{1 \leq i \leq 2} \left(\varphi_{\mathit{cross}_i} \land \varphi_{\mathit{add}_i}\right)$, where
\begin{align*}
	\varphi_\mathit{safe} &:= \Globally \neg \left( (\mathit{at\_crossing}_1 \land \Next \mathit{go}_1) \land (\mathit{at\_crossing}_2 \land \Next \mathit{go}_2) \right)\\
	\varphi_{\mathit{cross}_i} &:= \Globally \left( \mathit{at\_crossing}_i \rightarrow \Next \Eventually \mathit{go}_i \right)\\
	\varphi_{\mathit{add}_i} &:= m_i \land \Globally \left( m_i \rightarrow \left(\Next \neg m_i \land \Next^2 \neg m_i \land \dots \Next^{n_i-1} \neg m_i \land \Next^{n_i} m_i \right) \right)
\end{align*}
and where $\Next^x$ is syntactic sugar for applying $\Next$ $x$-times.

\subsection*{Detailed Results for the Manufacturing Robots}

\begin{table}[t]
    \centering
    \caption{Detailed results for the \emph{Manufacturing Robots} benchmark. Reported are the parameters, the implementation sizes of distributed BoSy and certifying synthesis, and the running time in seconds. We used a machine with a 3.1 GHz Dual-Core Intel Core i5 processor and 16 GB of RAM, and a timeout of 60 minutes. For distributed BoSy, the average runtime of 10 runs is given.\\}\label{table:detailed_results}

    \begin{tabular}{>{\centering}p{1.3cm}|>{\centering}p{1.9cm}>{\centering}p{1.9cm}||>{\centering}p{2.1cm}|>{\centering\arraybackslash}p{2.1cm}}
    	 & \multicolumn{2}{c||}{Strategy Size} & & \\
	     Param. & Cert. Synth. & Dist. BoSy & Cert. Synth. & Dist. BoSy \\
	     \hline\hline
	     0, 0 & 2, 2 & 2 & \textbf{1.10} & 2.45\\
	     \arrayrulecolor{black!40}\hline\arrayrulecolor{black}
	     2, 3 & 2, 6 & 6 & \textbf{1.59} & 2.91\\
	     2, 4 & 2, 4 & 4 & \textbf{1.18} & 2.43\\
	     2, 5 & 2, 10 & 10 & \textbf{3.97} & 299.11\\
	     2, 6 & 2, 6 & 6 & \textbf{1.40} & 3.25\\
	     2, 7 & 2, 14 & 14 & \textbf{76.32} & TO\\
	     2, 8 & 2, 8 & 8 & \textbf{2.47} & 5.28\\
	     2, 9 & 2, 18 & 18 & \textbf{1832.53} & TO\\
	     2, 10 & 2, 10 & 10 & \textbf{7.78} & 106.34\\
	     \arrayrulecolor{black!40}\hline\arrayrulecolor{black}
	     3, 4 & 6, 4 & 12 & \textbf{1.44} & TO\\
	     3, 5 & 6, 10 & 30 & \textbf{32.83} & TO\\
	     3, 6 & 6, 6 & 6 & \textbf{2.04} & 3.43\\
	     3, 7 & 6, 14 & 42 & \textbf{373.90} & TO\\
	     3, 8 & 6, 8 & 24 & \textbf{8.82} & TO\\
	     3, 9 & 6, 18 & 18 & TO & TO\\
	     3, 10 & 6, 10 & 30 & \textbf{30.92} & TO
    \end{tabular}
\end{table}

\begin{table}[t]
    \centering
    \caption{Detailed results for the \emph{Manufacturing Robots} benchmark. Reported are the parameters, the implementation sizes of distributed BoSy and certifying synthesis, and the running time in seconds. We used a machine with a 3.1 GHz Dual-Core Intel Core i5 processor and 16 GB of RAM, and a timeout of 60 minutes. For distributed BoSy, the average runtime of 10 runs is given.\\}\label{table:detailed_results_2}

    \begin{tabular}{>{\centering}p{1.3cm}|>{\centering}p{1.9cm}>{\centering}p{1.9cm}||>{\centering}p{2.1cm}|>{\centering\arraybackslash}p{2.1cm}}
    	 & \multicolumn{2}{c||}{Strategy Size} & & \\
	     Param. & Cert. Synth. & Dist. BoSy & Cert. Synth. & Dist. BoSy \\
	     \hline\hline
	     4, 5 & 4, 10 & 20 & \textbf{11.66} & TO\\
	     4, 6 & 4, 6 & 12 & \textbf{2.04} & TO\\
	     4, 7 & 4, 14 & 28 & \textbf{221.17} & TO\\
	     4, 8 & 4, 8 & 8 & \textbf{3.28} & 6.06\\
	     4, 9 & 4, 18 & 36 & \textbf{2911.26} & TO\\
	     4, 10 & 4, 10 & 20 & \textbf{7,93} & TO\\
	     \arrayrulecolor{black!40}\hline\arrayrulecolor{black}
	     5, 6 & 10, 6 & 30 & \textbf{26.16} & TO\\
	     5, 7 & 10, 14 & 35 & TO & TO\\
	     5, 8 & 10, 8 & 40 & \textbf{26.164} & TO\\
	     5, 9 & 10, 18 & 45 & TO & TO\\
	     5, 10 & 10, 10 & 10 & \textbf{89.87} & 335.98
    \end{tabular}
\end{table}

Lastly, we present the detailed experimental results for the \emph{Manufacturing Robots} benchmark.
The benchmark describes the robots from \Cref{sec:motivating_example}. It is parameterized in the additional objectives $\varphi_{\mathit{add}_i}$ of the robots. Hence, the interface between the processes, \ie, their certificates, stay small while the size of the strategy increases.
The specification has two parameters, $n_1$ and $n_2$. The additional objectives of the robots state the $r_i$ needs to visit the machine it is responsible for in every $n_i$-th step. 
The smallest certificates allowing for satisfying the $\varphi_\mathit{safe}$ and $\varphi_{\mathit{cross}_i}$ are of size two: The robots take turns in who is allowed to enter the crossing.
Together with the scalable additional requirements, this results in different minimal strategy sizes of the robots for certifying synthesis and distributed BoSy: While the size of the solutions of certifying synthesis only depends on the size of the certificate, which is two for all parameters, and the parameter of the respective robot, the size of the solution of distributed BoSy depends on the certificate and the parameters for \emph{both} robots. Therefore, the sizes of the solutions and thus the synthesis times of both approaches do not grow in parallel.

The detailed experimental results for the \emph{Manufacturing Robots} benchmark are given in \Cref{table:detailed_results,table:detailed_results_2}. We report on the parameters $n_1$ and $n_2$ as well as on the strategy sizes (in terms of states of the strategy transition system) and running times in seconds for both certifying synthesis and distributed BoSy. Since there do not exist dominant strategies for this benchmark, we omitted this column.
Note that the strategy size of distributed BoSy serves as the parameter in \Cref{table:results}.
For benchmarks resulting in the same strategy size, the average runtime is depicted in \Cref{table:results}. For instance the benchmarks with parameters (2,3), (2,6), and (3,6) all result in strategy size $6$ for distributed BoSy and hence the average of their running times is given for parameter $6$ for both certifying synthesis and distributed BoSy in \Cref{table:results}.

Clearly, the running times of both certifying synthesis and distributed BoSy depend highly on the (minimal) size of the synthesized strategy.
Although the sizes are similar if $n_1 = 2$, certifying synthesis clearly outperforms distributed BoSy there. For $n_1 > 2$, the advantage of certifying synthesis is even bigger: Since focusing on the certificates instead of on the strategies of the other processes, allows for abstracting from irrelevant behavior. Hence, the sizes of the minimal strategies of solutions of certifying synthesis are significantly smaller than the ones of distributed BoSy. Hence, while distributed BoSy reaches the timeout of 60 minuted for almost all instances with $n_1 > 2$, certifying synthesis yields solutions for all but one instance, for most instances even in less than one minute.

%% file: constraint_system.tex
Next, we present the SAT constraint system $\mathcal{C}_{A,\varphi}$ that, given an architecture~$A$ and an LTL formula $\varphi$, encodes the search for \emph{local strategies} and \emph{guarantee transition systems} satisfying the requirements of (approximative) certifying synthesis with local strategies for $\varphi$. That is, we present the desired SAT constraint system of \Cref{thm:constraint_system}.
Intuitively, our constraint system $\constraintSystem{A}{\varphi}{\mathcal{B}}$ consists of $n$ copies of the constraint system for monolithic systems~\cite{FaymonvilleFRT17}, one for each process, to search for the strategies~$s_j$, and some further constraints for searching for the certificates and for ensuring the correct relation between certificates and strategies.

Let $\varphi$ be an LTL specification. Recall that in bounded synthesis~\cite{FinkbeinerS13}, $\varphi$ is translated into a universal co-Büchi automaton $\mathcal{A}$ that accepts $\mathcal{L}(\varphi)$, \ie, with $\mathcal{L}(\mathcal{A}) = \mathcal{L}(\varphi)$. 
A transition system $\mathcal{T}$ is accepted by $\mathcal{A}$ if for every input sequence, all runs of $\mathcal{A}$ induced by the path of $\mathcal{T}$ on the input are accepting, \ie, if they only visit finitely many rejecting states.
To consider the runs of $\mathcal{A}$ induced by the paths of $\mathcal{T}$ on all inputs, we build the (unique) \emph{run graph} of $\mathcal{A}$ and $\mathcal{T}$ and check whether all of its paths have only finitely many visits to rejecting states.

\begin{definition}[Run Graph]
	Let $\mathcal{T}=(T,t_o,\tau,o)$ be a Moore transition system and let $\mathcal{A}=(Q,q_0,\delta,F)$ be a universal co-Büchi automaton. The unique \emph{run graph} $\mathcal{G} = (V,E)$ of $\mathcal{T}$ and $\mathcal{A}$ is defined by $V = T \times Q$ and $((t,q),(t',q')) \in E$ if, and only if, there is a valuation $\inp{i}\in 2^I$ of the inputs and a valuation $\out{i} \in 2^O$ of the outputs such that $\tau(t,\inp{i}) = t'$, $o(t)=\out{o}$, and $(q,\inp{i}\cup\out{o},q')\in\delta$ hold.
\end{definition}

To check whether all of the run graph's paths have only finitely many visits to rejecting states, we \emph{annotate} it. An \emph{annotation} $\lambda: T \times Q \rightarrow \mathbb{N} \cup \{\bot\}$ maps nodes of the run graph to either unreachable ($\bot$), or to a natural number $k$. Intuitively, the annotation counts the number of visits to rejecting states that already can have occurred when reaching a state. A path with infinitely many visits to rejecting states would require a state to be annotated with $\infty$.

\begin{definition}[Valid Annotation]
	Let $\mathcal{T}=(T,t_0,\tau,o)$ be a Moore transition system and let $\mathcal{A}=(Q,q_0,\delta,F)$ be a universal co-Büchi automaton. An annotation $\lambda: T \times Q \rightarrow \mathbb{N} \cup \{\bot\}$ is called \emph{valid} if
	\begin{enumerate}
		\item The pair of initial states $(t_0,q_0)$ is annotated with a natural number, \ie we have $\lambda(t_0,q_0) = k \neq \bot$.
		\item If a pair of states $(t,q)$ is annotated with a natural number, then every successor pair $(t',q')$ in the run graph is labeled with a greater number. The number has to be strictly greater if $q'$ is rejecting, \ie, we require $\lambda(t',q') \greaterBound{q'} \lambda(t,q)$, where $\greaterBound{q'} := >$ if $q' \in F$ and $\greaterBound{q'} := \geq$ if $q' \not\in F$.
	\end{enumerate}
\end{definition}

The existence of a valid annotation indeed corresponds to the satisfaction of the specification by the transition system:

\begin{theorem}[\cite{FinkbeinerS13}]\label{thm:valid_annotation}
	A transition system $\mathcal{T}$ is accepted by a universal co-Büchi automaton $\mathcal{A}$ if, and only if, it has a valid annotation.
\end{theorem}

Using valid annotations, we define the constraint system.
First, we define the encodings of the transition systems representing the strategies and certificates as well as of the universal co-Büchi automata representing the subspecifications and the annotations of the resulting run graphs.
Let $\mathcal{T}_j = (T_j,t^j_0,\tau_j,o_j)$ be the TS representing strategy~$s_j$ of a process $p_j\in\sysProc$. Let $\guarTrans{j} = (G_j,u^j_0,\tau^G_j,o^G_j)$ be the GTS representing certificate $g_j$ for $p_j \in \sysProc$. Let $\mathcal{A}_j = (Q_j,q^j_0,\delta_j,F_j)$ be the universal co-Büchi automaton representing the subspecification $\varphi_j$ for $p_j \in \sysProc$.
We encode them as well as the annotations $\lambda^\mathbb{B}_j: T_j \times Q_j \rightarrow \mathbb{B}$ encoding reachability and $\lambda^{\#}_j: T_j \times Q_j \rightarrow \mathbb{N}$ encoding the natural number as follows:

\begin{itemize}
	\item Universal co-Büchi Automaton $\mathcal{A}_j = (Q_j,q^j_0,\delta_j,F_j)$:
		\begin{itemize}
			\item[-] $\transSpec{j}{q}{\inp{i}}{q'}$ is true iff $(q,\inp{i},q') \in \delta_j$ 
		\end{itemize}
	\item Strategy Transition System $\mathcal{T}_j = (T_j, t^j_0, \tau_j, o_j)$: 
		\begin{itemize}
			\item[-] $\transStrat{j}{t}{\inp{i}}{t'}$ is true iff $\tau_j(t,\inp{i}) = t'$
		 	\item[-] $\outputStrat{j}{t}{\inp{i}}{v}$ is true iff $o_j(t) = L$ and $v \in L$
		\end{itemize}
	\item Guarantee Transition System $\guarTrans{j} = (G_j, u^j_0, \tau^G_j, o^G_j)$:
		\begin{itemize}
			\item[-] $\transGuar{j}{u}{\inp{i}}{u'}$ is true iff $\tau^G_j(u,\inp{i}) = u'$
		 	\item[-] $\outputGuar{j}{u}{\inp{i}}{v}$ is true iff $o^G_j(u) = L$ and $v \in L$
		\end{itemize}
	\item Annotation $\lambda$:
		\begin{itemize}
			\item[-] $\reach{j}{t}{q}$ is true iff $(t,q)$ is reachable from the initial state of the run graph
			\item[-] $\bound{j}{t}{q}$ is the bit vector of length $\mathcal{O}(\log(|T|\cdot|Q|))$ representing the binary encoding of the value $\lambda(t,q)$ 
		\end{itemize}
\end{itemize}

We now define the constraint system $\constraintSystem{A}{\varphi}{\mathcal{B}}$ for an architecture $A$, an LTL formula $\varphi$ represented by a universal co-Büchi automaton $\mathcal{A}$, and size bounds~$\mathcal{B}$. We give the constraints for a process $p_j \in \sysProc$, all in all, we obtain:

\[ \bigwedge_{p_j \in \sysProc} \left( (a) \land (b) \land \bigwedge_{p_k \in \relevantProcesses{j}} (c) \land (d) \land (e) \right). \]

\noindent
(a) Guarantee transition systems are required to be complete. That is, they have exactly one outgoing transition for every state and every input:

\noindent\scalebox{0.98}{\parbox{1.020408\linewidth}{%
\begin{equation}\tag{a}
	\begin{split}
	\bigwedge_{u\in\statesGuar{j}} \bigwedge_{\inp{i} \in \inputs{j}} \left( \bigvee_{u'\in\statesGuar{j}} \transGuar{j}{u}{\inp{i}}{u'} \right) \land \bigwedge_{u'\in\statesGuar{j}} \bigwedge_{\substack{u''\in\statesGuar{j} \\ u' \neq u''}} \left(\neg \transGuar{j}{u}{\inp{i}}{u'} \lor \neg \transGuar{j}{u}{\inp{i}}{u''}  \right)
	\end{split}
\end{equation}
}}

\noindent
(b) A local strategy needs to satisfy its own certificate. In fact, certifying synthesis ensures that $s_j \simresp{\guarOutputs{j}} g_j$ holds. Thus, we encode the existence of a simulation relation $\simRelStG{j}: T_j \times G_j$ from $s_j$ to $g_j$. To do so, we introduce a variable $\simStratToGuar{j}{t}{u}$ that is true iff $(t,u) \in \simRelStG{j}$ and obtain:

\noindent\scalebox{0.98}{\parbox{1.020408\linewidth}{%
\begin{equation}\tag{b}
	\begin{split}
	\simStratToGuar{j}{t^j_0}{u^j_0} &\land \bigwedge_{t\in\statesStrat{j}} \bigwedge_{u\in\statesGuar{j}} \left( \simStratToGuar{j}{t}{u} \rightarrow  \left( \bigwedge_{v\in\guarOutputs{j}} \left( \outputStrat{j}{t}{\inp{i}}{v} \leftrightarrow \outputGuar{j}{u}{\inp{i}}{v} \right) \land \right.\right. \\
	& \left.\left. \bigwedge_{\inp{i}\subseteq\inputs{j}}\bigwedge_{t'\in\statesStrat{j}} \left( \transStrat{j}{t}{\inp{i}}{t'} \rightarrow \bigvee_{u'\in\statesGuar{j}} \left( \transGuar{j}{u}{\inp{i}}{u'} \land \simStratToGuar{j}{t'}{u'} \right) \right)\!\!\right)\!\!\right)
	\end{split}
\end{equation}
}}

\noindent
(c) A local strategy only needs to satisfy the specification if the other processes stick to their guarantees. In certifying synthesis with local strategies, we ensure this by representing strategies by incomplete transition systems whose computations are infinite if, and only if, the other (relevant) processes do not deviate from their certificates. We did this with the notion of valid histories. To encode this in our constraint system, the main idea is to assign a set of \emph{associated outputs}~$\associatedOutputs{j}$ to each process $p_j \in\sysProc$. The associated outputs are the output variables of $p_j$'s relevant processes, \ie, $\associatedOutputs{j} = \{ \outputs{k} \mid p_k \in \relevantProcesses{j}\}$. The labeling function of the strategy transition system is then defined over $\outputs{j} \cup \associatedOutputs{j}$. We introduce a constraint that ensures that the valuations of the associated outputs indeed match the certificates of the relevant processes. In fact, we require that for all relevant processes $p_k \in \relevantProcesses{j}$, $s_j$ simulates $g_k$ regarding the associated outputs of $s_j$ that are outputs of $g_k$. Note here that we use a slightly more general definition of simulation than in the preliminaries: The initial states need to be contained in the simulation relation $\simRelGtS{k}{j}$, \ie, $(u^k_0,t^j_0) \in \simRelGtS{k}{j}$. Furthermore, the successors $u'$ and $t'$ of states $u$ and $t$ with $(u,t) \in \simRelGtS{k}{j}$ for inputs $\inp{i} \in 2^\inputs{k}$ and $\inp{i}'\in\inputs{j}$ that agree on shared variables and that may occur if the other processes stick to their guarantees need to be contained in $\simRelGtS{k}{j}$. Lastly, if $(u,t) \in \simRelGtS{k}{j}$ holds, then they need to agree on the associated outputs of $p_j$ that are outputs of $p_k$. 
Note that $g_k$ is a complete guarantee transition system, while $s_j$ is an incomplete local strategy (as we will encode later in constraint (d)). Hence, not every transition of $g_k$ can be matched by a transition of $s_j$: If the input sequence does not respect the certificate of a further relevant process~$p_\ell$, then there is no transition in~$s_j$. Therefore, we only consider input sequences that may occur if the other processes stick to their guarantees in the definition of the simulation. For the sake of readability, we encode the check whether an input sequence may occur with $\validInput{j}{t}{\inp{i}'}:=\bigwedge_{v \in \associatedOutputs{j}} (v\in\inp{i}' \leftrightarrow \outputStrat{j}{t}{\inp{i}'}{v})$. We obtain the following constraint:

\noindent\scalebox{0.98}{\parbox{1.020408\linewidth}{%
\begin{align}\tag{c}
	&\simGuarToStrat{k}{j}{u^k_0}{t^j_0} \land \bigwedge_{u\in\statesGuar{k}} \bigwedge_{t\in\statesStrat{j}} \left( \simGuarToStrat{k}{j}{u}{t} \rightarrow \bigwedge_{\inp{i}\subseteq\inputs{k}} \bigwedge_{\substack{\inp{i'}\subseteq\inputs{j} \\ \inp{i}\cap\inputs{j} = \inp{i'}\cap\inputs{k}}} \left( \validInput{j}{t}{\inp{i'}} \rightarrow \right.\right.\\
	&~~~~~~\left.\left.\left(\bigwedge_{u'\in\statesGuar{k}} \left( \transGuar{k}{u}{\inp{i}}{u'} \rightarrow \bigvee_{t'\in\statesStrat{j}} \left( \transStrat{j}{t}{\inp{i'}}{t'} \land \simGuarToStrat{k}{j}{u'}{t'} \right)\!\! \right) \land \!\! \bigwedge_{v\in\associatedOutputs{j}\cap\guarOutputs{k}} \!\! \outputGuar{k}{u}{\inp{i}}{v} \leftrightarrow \outputStrat{j}{t}{\inp{i'}}{v} \right)\!\! \right)\!\! \right)\nonumber
\end{align}
}}

\noindent
To ensure that $s_j$ simulates not only a single certificate $g_k$ but \emph{all} relevant certificates, we use $\bigwedge_{p_k \in \relevantProcesses{j}} (c)$ in the overall constraint system.

\noindent
(d) A local strategy is complete for all inputs that may occur if the other (relevant) processes stick to their certificates and incomplete for all other inputs. That is, for every state, they have exactly one outgoing transition for every input that may occur if the other (relevant) processes stick to their certificates, and no outgoing transitions for every other input. Since by constraint (c), the associated outputs of $p_j$ in $s_j$ exactly capture the outputs of the other (relevant) processes that may occur at this particular point in time, we can again use $\validInput{j}{t}{\inp{i}'}$ to determine whether the existence of an outgoing transition is required:

\noindent\scalebox{0.98}{\parbox{1.020408\linewidth}{%
\begin{equation}\tag{d}
	\begin{split}
	\bigwedge_{t\in\statesStrat{j}} \bigwedge_{\inp{i} \in \inputs{j}} \left( \validInput{j}{t}{\inp{i}} \leftrightarrow \bigvee_{t'\in\statesStrat{j}} \transStrat{j}{t}{\inp{i}}{t'} \right) \land \bigwedge_{t'\in\statesStrat{j}} \bigwedge_{\substack{t''\in\statesStrat{j} \\ t' \neq t''}} \left(\neg \transStrat{j}{t}{\inp{i}}{t'} \lor \neg \transStrat{j}{t}{\inp{i}}{t''}  \right)
	\end{split}
\end{equation}
}}

\noindent
(e) Lastly, the run graph of the strategy transition system and the universal co-Büchi automaton is required to have a valid annotation:

\noindent\scalebox{0.98}{\parbox{1.020408\linewidth}{%
\begin{align}\tag{e}
	\reach{j}{t^j_0}{q^j_0} &\land \bigwedge_{q\in\statesSpec{j}} \bigwedge_{t\in\statesStrat{j}} \left( \reach{j}{t}{q} \rightarrow \right. \\
	&\left.\bigwedge_{q'\in\statesSpec{j}} \bigwedge_{\inp{i}\subseteq\inputs{j}} \left( \bigwedge_{\out{o} \subseteq \outputs{j}} \transSpecSugar{j}{t}{q}{\inp{i} \cup \inp{o}}{q'} \rightarrow \bigwedge_{t'\in\statesStrat{j}} \left( \transStrat{j}{t}{\inp{i}}{t'} \rightarrow \left( \reach{j}{t'}{q'} \land \bound{j}{t'}{q'} \greaterBound{q'} \bound{j}{t}{q} \right) \!\! \right) \!\! \right) \!\! \right)\nonumber
\end{align}
}}

where $\transSpecSugar{j}{t}{q}{\inp{i} \cup \inp{o}}{q'}$ is syntactic sugar for $ \transSpec{j}{q}{\inp{i} \cup \out{o}}{q'} \land \bigwedge_{o\in\outputs{j}} \outputStrat{j}{t}{\inp{i}}{o} \leftrightarrow o \in \out{o}$, \ie, $\transSpecSugar{j}{t}{q}{\inp{i} \cup \inp{o}}{q'}$ is true, iff $o_j(t) = \inp{o}$ and $(q,\inp{i}\cup\inp{o},q')\in\delta_j$ hold.

To \emph{search} for local strategies $s_1, \dots, s_n$ and GTS $g_1, \dots, g_n$, we quantify existentially over the variables to find an assignment. If the constraint system is realizable, then the solution defines a solution of certifying synthesis with local strategies. Otherwise, the specification is unrealizable for the given size bounds.

The correctness of the constraint system follows with \Cref{thm:corectness_local_strategies,thm:valid_annotation}, the correctness of the constraint system for monolithic systems as well as the fact that using associated outputs and ensuring that they match the certificates indeed corresponds to checking for the existence of valid histories.

Note that we represent strategies with Moore transition systems. This ensures that the parallel composition of strategies is again a complete transition system. For Mealy transition systems, \ie, transition systems where the labeling may depend on the state \emph{and the input}, this does not hold in general. However, we can extend certifying synthesis and, in particular, the constraint system presented above: 
First, the variables encoding the labeling function need to depend on the input as well (c.f.\ the constraint system for monolithic systems~\cite{FinkbeinerS13,FaymonvilleFRT17} that also works for Mealy transition systems). 
Second, for Mealy transition systems, it is trivial to satisfy the righthand side of the implication in constraint (c) by violating $\validInput{j}{t}{\inp{i}'}$: The strategy reacts directly to the input and sets an associated output if, and only if, it is contained in the input. Thus, when considering Mealy transition systems, we need to ensure that such a solution is not possible. Therefore, we include a constraint describing that every strategy needs to have at least one transition for every environment output:
\begin{equation*}
	\bigwedge_{t\in\statesStrat{j}} \bigwedge_{\inp{i} \subseteq \envOutputs} \bigvee_{\inp{i'} \in \associatedOutputs{j}} \bigvee_{t'\in\statesStrat{j}} \transStrat{j}{t}{\inp{i} \cup \inp{i'}}{t'}
\end{equation*}
Together with constraint (d), this ensures that violating $\validInput{j}{t}{\inp{i}'}$ for every input sequence is not possible.
Furthermore, this immediately implies that the parallel composition of the strategy transition systems is complete.